\documentclass[journal]{IEEEtran}

\usepackage{amsmath,amssymb,graphicx,epsf,psfrag,epsfig,cite,color}
\usepackage[left=1in,top=0.5in,right=1in,bottom=0.6in]{geometry}
\usepackage{verbatim}
\usepackage[font=footnotesize,skip=4pt]{caption}

\def\etal{{\it et al. \/}}
\def\ie{i.e., }
\def\scalefig#1{\epsfxsize #1\textwidth}
\def\defeq{{\stackrel{\Delta}{=}}}

\newtheorem{theorem}{\bf Theorem}
\newtheorem{lemma}{\bf Lemma}

\newtheorem{define}{\bf Definition}
\newcommand{\BigO}[1]{\ensuremath{\operatorname{O}\bigl(#1\bigr)}}

\newcounter{algleo}
\newlength{\lefttab}
  {\trivlist
   \topsep=0pt\itemsep=0pt
   \addtolength{\lefttab}{1.25em}
   \leftskip=\lefttab}%
  {\endtrivlist}

\renewcommand\appendices{\par
  \setcounter{section}{0}
  \setcounter{subsection}{0}
  \renewcommand\thesection{Appendix~\Alph{section}}
}

\makeatletter
\def\@maketitle{%
\newpage
\null
\vskip 0em
\begin{center}%
\let \footnote \thanks
{\LARGE \@title \par}%
\vskip 1em
{\large
\lineskip .5em%
\begin{tabular}[t]{c}%
\@author
\end{tabular}\par}%
\end{center}%
\par
\vskip -3em}
\makeatother

\newcommand{\SectionVspace}{-0.5em}
\newcommand{\SubsectionVspace}{-0.8em}

\usepackage{etoolbox}
\newcommand{\zerodisplayskips}{%
  \setlength{\abovedisplayskip}{0.3pt}
  \setlength{\belowdisplayskip}{0.3pt}}
\appto{\normalsize}{\zerodisplayskips}
\appto{\small}{\zerodisplayskips}
\appto{\footnotesize}{\zerodisplayskips}

\begin{document}

\title{The Thinnest Path Problem}

\author{Jianhang Gao$^\dag$, Qing Zhao$^\dag$, Ananthram Swami$^\S$\\
$^\dag$University of California, Davis, CA 95616, \{jhgao,qzhao\}@ucdavis.edu\\
$^\S$Army Research Laboratory, Adelphi, MD 20783, a.swami@ieee.org}

\maketitle

\begin{abstract}
We\footnotetext{This work was supported by the Army Research Laboratory Network Science CTA under Cooperative Agreement W911NF-09-2-0053.} formulate and study the thinnest path problem in wireless ad hoc networks. The objective is to find a path from a source to its destination that results in the minimum number of nodes overhearing the message by a judicious choice of relaying nodes and their corresponding transmission power. We adopt a directed hypergraph model of the problem and establish the NP-completeness of the problem in 2-D networks. We then develop two polynomial-time approximation algorithms that offer $\sqrt{\frac{n}{2}}$ and $\frac{n}{2\sqrt{n-1}}$ approximation ratios for general directed hypergraphs (which can model non-isomorphic signal propagation in space) and constant approximation ratios for ring hypergraphs (which result from isomorphic signal propagation). We also consider the thinnest path problem in 1-D networks and 1-D networks embedded in 2-D field of eavesdroppers with arbitrary unknown locations (the so-called 1.5-D networks). We propose a linear-complexity algorithm based on nested backward induction that obtains the optimal solution to both 1-D and 1.5-D networks. This algorithm does not require the knowledge of eavesdropper locations and achieves the best performance offered by any algorithm that assumes complete location information of the eavesdroppers. 
\end{abstract}

\vspace{-1em}
\section{Introduction}
\subsection{The Thinnest Path Problem}
\label{sec:intro-thinnest}
We consider the \emph{thinnest path} problem in wireless ad hoc networks. For a given source and a destination, the thinnest path problem
asks for a path from the source to the destination that results in the minimum number of nodes hearing the message. Such a path
is achieved by carefully choosing a sequence of relaying nodes and their corresponding transmission power.

At the first glance, one may wonder whether the thinnest path problem is simply a shortest path problem with the weight of each hop given by the number of nodes that hear the message in this hop. Realizing that a node may be within the transmission range of multiple relaying
nodes and should not be counted multiple times in the total weight (referred to as the width) of the resulting path, we see that the thinnest
path problem does not have a simple cost function that is summable over edges. But rather, the width of a path is given by the cardinality of
the \emph{union} of all receiving nodes in each hop, which is a highly nonlinear function of the weight of each hop. One may then wonder whether
we can redefine the weight of each hop as the number of nodes that hear the message for the first time. Such a definition of edge weight indeed leads
to a summable cost function. Unfortunately, in this case, the edge weight cannot be predetermined until the thinnest path from the source to
the edge in question has already been established.

A more fundamental difference between the thinnest path and the shortest path problems is that the thinnest path from a single source to all other nodes in the network do not form a tree. In other words, the thinnest path to a node does not necessarily go through the thinnest path to any of its neighbors. The loss of the tree structure is one of the main reasons that the thinnest path problem is much more complex than the shortest path problem. Indeed, as shown in this paper, the thinnest path problem is NP-complete, which
is in sharp contrast with the polynomial nature of the shortest path problem.

Another aspect that complicates the problem is the choice of the transmission power at each node (within a maximum value that may vary
across nodes). In this case, the network cannot be modeled as a simple graph in which the neighbors of each node are prefixed. In this paper,
we adopt the \emph{directed hypergraph} model which easily captures the choice of different neighbor sets at each node. While a graph is given by
a vertex set $V$ and an edge set $E$ consisting of cardinality-2
subsets of $V$, a hypergraph~\cite{berge1976graphs} is free of
the constraint on the cardinality of an edge. Specifically,
any non-empty subset of $V$ can be an element (referred to as a hyperedge) of the edge set
$E$. Hypergraphs can thus capture group behaviors and higher-dimensional
relationships in complex networks that are more than a simple union
of pairwise relationships\cite{Ramanathan&Etal:11INFORCOM}. In a directed hypergraph~\cite{Gallo&Etal:93DAM}, each hyperedge is directed, going from
a source vertex to a non-empty set of destination vertices. An example is given in Fig.~\ref{fig:hypergraph}-(a) where we have $2$ directed hyperedges rooted at a source node $v$ with each hyperedge modeling a neighbor set of $v$ under a specific power. The directed hypergraph
model of the thinnest path problem is thus readily seen: rooted at each node are
multiple directed hyperedges, each corresponding to a distinct neighbor set feasible under the maximum transmission power
of this node. The problem is then to find a minimum-width hyperpath from the source to the destination where the width of a hyperpath is given by the cardinality
of the union of the hyperedges on this hyperpath.

\vspace{\SubsectionVspace}
\subsection{Main Results}
Based on the directed hypergraph formulation, we show that the thinnest path problem in $2$-D networks is NP-complete even under a simple disk propagation model. This result is established through a reduction from the minimum dominating set problem in graphs, a classic NP-complete problem. The most challenging part of this reduction is to show the reduced problem is realizable under a 2-D disk model. We further establish that even with a fixed transmission power at each node (in this case, the resulting hypergraph degenerates to a standard graph), the thinnest path problem is NP-complete. We then propose two polynomial-time approximation algorithms that offer $\sqrt{\frac{n}{2}}$ and $\frac{n}{2\sqrt{n-1}}$ approximation ratios for general directed hypergraphs (which can model non-isomorphic signal propagation
in space) and constant approximation ratios for ring hypergraphs (which result from isomorphic signal propagation). 

We also establish the polynomial nature of the problem in 1-D and 1.5-D networks,where a 1.5-D network is a 1-D network embedded in a 2-D field of eavesdroppers with arbitrary unknown locations. We propose an algorithm based on a nested backward induction (NBI) starting at the destination. We show that this NBI algorithm has $\BigO{n}$ time complexity. Since the size of the input date is $\BigO{n}$, the proposed algorithm is order optimal. It solves the thinnest path problem in both the 1-D and 1.5-D networks. In particular, no algorithm, even with the complete location information of the eavesdroppers, can obtain a thinner path than the NBI algorithm which does not require the knowledge of eavesdropper locations.

\vspace{\SubsectionVspace}
\subsection{Related Work}
There is a large body of literature on security issues in wireless ad hoc networks (see, for example, \cite{Anjum&Mouchtaris:book,Hu04}).
However, the thinnest path problem has not been studied in the literature except in \cite{Chechik&etal2012}. Chechik \etal studied the thinnest path (referred to as the secluded path in~\cite{Chechik&etal2012}) and the thinnest Steiner tree in graphs. They showed that the problem in a general graph is NP-complete and strongly inapproximable. They proposed an approximation algorithm with an approximation ratio of $\sqrt{\Delta}+3$ where $\Delta$ is the maximum degree of the graph. They further studied the problem in several special graph models including graphs with bounded degree, hereditary graphs, and planar graphs. However, their study focuses on the problem in topological graphs, whereas we focus on hypergraphs and geometric graphs.
 The results obtained in~\cite{Chechik&etal2012} do not apply to special hypergraphs satisfying certain geometric properties that result naturally from the communication problem studied in this paper. This paper also includes several new results on the thinnest path problem under the graph model. Specifically, we established the NP-completeness of the problem in $2$-D disk graphs and $3$-D unit disk graphs.  The results in~\cite{Chechik&etal2012} and this work thus complement each other to provide a more complete picture of the thinnest path problem under different (hyper)graph models.

In the general context of algorithmic studies in hypergraphs, Ausiello~\etal\cite{ausiello1992optimal} tackled the problem of finding the $\mu$-optimal hyperpath where $\mu$ is a general measure on hyperpaths that satisfies a certain monotone property. They established the NP-completeness of this problem for general measures. The thinnest path problem can be seen as an $\mu$-optimal traversal problem with the measure $\mu$ given by the number of vertices covered by the path. Since it is a special measure, their NP-completeness result developed under general measures does not apply. Furthermore, in many applications, the resulting hypergraphs have certain topological and/or geometrical properties, and the computational complexities under these special models require separate analysis.

Another realated problem is the shortest path problem in hypergraphs. The shortest path problem in hypergraphs remains a polynomial-time problem as its counterpart under the graph model. The static shortest hyperpath problem was considered by Knuth~\cite{Knuth:77IPL} and Gallo \etal~\cite{Gallo&Etal:93DAM} in which Dijkstra's algorithm for graph was
extended to obtain the shortest hyperpaths. Ausiello \etal proposed a dynamic shortest hyperpath algorithm for directed hypergraphs, considering only the incremental problem (\ie network changes contain only edge insertion and weight decrease) with the weights of all hyperedges limited to a finite
set of numbers~\cite{ausiello1990dynamic}. In~\cite{Gao&etal:12WiOpt}, Gao \etal developed the first fully dynamic shortest path algorithms for general hypergraphs. As discussed in ~\ref{sec:intro-thinnest}, the thinnest path problem is fundamentally different and significantly more complex than the shortest path problem.

The widest path problem has been well studied under the graph model\cite{Hu1961,Punnen91}, and the existing results can be easily extended to hypergraphs. The widest path problem asks for a path whose minimum edge weight along the path is maximized. In other words, the width of a path is given by the minimum edge weight on that path, which is different from the definition of path width in the thinnest path problem studied in this paper. As a consequence, the widest path problem is not the complement of the thinnest path problem. Since the tree structure is preserved in the widest path problem (\ie the widest path to a node must go through the widest path to one of its neighbors), it remains a polynomial time problem. The thinnest path problem, however, is NP-complete in general.

\vspace{\SectionVspace}
\section{Problem Formulation}\label{sec:formu}
\subsection{Basic Concepts of Directed Hypergraphs}

A \emph{directed hypergraph} $H=(V,E)$ contains a set $V$ of vertices and a set $E$ of directed hyperedges~\cite{Gallo&Etal:93DAM}\footnote{In \cite{Gallo&Etal:93DAM}, it is referred to as the forward hyperarcs.}. Each directed hyperedge $e\in E$ has a single source vertex $s_e$ and a non-empty set of destination vertices $T_e$.

\begin{figure}[htbp]
\centering
\begin{psfrags}
\psfrag{a}[c]{(a)}
\psfrag{b}[c]{(b)}
\psfrag{c}[c]{(c)}
\scalefig{0.45}\epsfbox{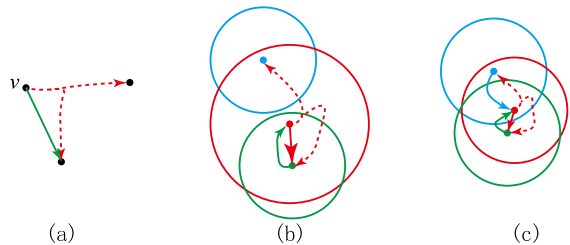}
\end{psfrags}
\caption{(a): Directed hypergraph; (b): Disk hypergraph; (c): Unit disk hypergraph.}
\label{fig:hypergraph}
\end{figure}

A \emph{disk hypergraph} is a special directed hypergraph whose topology is determined by a set of points $V=\{v_1,\ldots,v_n\}$ located in a $d$-dimensional Euclidean space and a maximum range $R_i$ associated with each vertex $v_i$. There exists a hyperedge $e$ from source $s_e=v_i$ to destination set $T_e$ if and only if $T_e$ consists of vertices located within the $d$-dimensional sphere centered at $v_i$ with a radius $r\in (0,R_i]$. A \emph{unit disk hypergraph} (UDH) is a disk hypergraph with unit maximum range ($R_i=1$) for all vertices. Fig.~\ref{fig:hypergraph} shows examples of a directed hypergraph, a disk hypergraph, and a unit disk hypergraph. 

A \emph{ring hypergraph} is a generalized disk hypergraph where associated with each vertex $v_i$ is a minimum range $r_i$ as well as a maximum range $R_i$. Hyperedges rooted at $v_i$ are formed by spheres centered at $v_i$ with radiuses satisfying $r_i< r\leq R_i$. It is easy to see that a disk hypergraph is a ring hypergraph with $r_i=0$, a disk \emph{graph} is a ring hypergraph with $r_i=R_i$ for all $i$, and a unit disk graph (UDG) is a ring hypergraph with $r_i=R_i=1$ for all $i$. 

\vspace{\SubsectionVspace}
\subsection{The Thinnest Path Problem }
Consider a wireless ad-hoc network with $n$ nodes located in a $d$-dimension Euclidean space. Each node can choose the power, within a maximum value, for the transmission of each message. The chosen power, along with the signal propagation model, determines the set of neighbors that can hear the message. The maximum transmission power is in general different across nodes. The objective is to find a path between a given source-destination pair that involves the minimum number of nodes hearing the message.

As discussed in Sec.~\ref{sec:intro-thinnest}, we formulate the problem using a directed hypergraph. Each node is a vertex. The directed hyperedges rooted at a node are given by distinct neighbor sets of this node feasible under its maximum transmission power and the signal propagation model. Under a general nonisomorphic propagation model, we end up with a general hypergraph. The only property the resulting hypergraph has is the monotonicity of the hyperedge set. Specifically, the hyperedges rooted at each node can be ordered in such a way (say, $e_1,e_2,\ldots,e_l$) that $T_{e_i}\subset T_{e_i+1}$ and $|T_{e_i}|=|T_{e_{i+1}}|-1$. This is due to the nature of wireless broadcasting where nodes reachable under transmission power $\eta$ can also be reached under any power greater than $\eta$. Under an isomorphic propagation model, we end up with a disk hypergraph. If all nodes have the same maximum range\footnote{Transmission range and transmission power are used interchangeably.}, we have a unit disk hypergraph. This hypergraph model also applies to networks with eavesdroppers. Each eavesdropper can be seen as a node with zero transmission range. It is thus a vertex with no outgoing hyperedges.

Given a source-destination pair $(s,t)$, a hyperpath from $s$ to $t$ is defined as a sequence of hyperedges $L=\{e_1,\ldots,e_m\}$ such that $s_{e_i}\in T_{e_{i-1}}$ for $1<i\leq m$, $s_{e_1}=s$, and $t\in T_{e_m}$. Define the cover $\widehat{L}$ of $L$ to be the set of vertices in $L$, \ie
\begin{align*}
\widehat{L}\defeq \cup_{i=1}^k T_{e_i}\cup \{s_{e_1}\},
\end{align*}
The width $W(L)$ is then given by
\begin{align*}
W(L)\defeq|\widehat{L}|.
\end{align*}
The thinnest path problem asks for a hyperpath from $s$ to $t$ with the minimum width. Note that choosing a hyperedge $e=\{s_e, T_e\}$
simultaneously chooses the relaying node $s_e$ and its transmission power (determined by $T_e$).

\vspace{\SectionVspace}
\section{NP-Complete Problems}\label{sec:npc}
In this section, we show that the thinnest path (TP) problem is NP-complete in several special geometric hypergraphs and graphs. This implies the NP-completeness of the problem in general directed hypergraphs. 

\vspace{\SubsectionVspace}
\subsection{TP in $2$-D Disk Hypergraphs}\label{sec:npc-2ddh}
In this subsection, we prove the NP-completeness of the thinnest path problem in $2$-D disk hypergraphs. While a stronger result is shown in the next subsection, the proof of this result provides the main building block for the proof of the next result.

The result is established through a reduction from the maximum dominating set (MDS)\cite{Garey&Johnson1979} problem. The MDS problem asks for the minimum subset of vertices in a given graph such that every vertex in the graph is either in the subset or a direct neighbor of one vertex in the subset. The following theorem formally establishes the polynomial reduction (denoted by $\leq_P$) from MDS to TP in $2$-D disk hypergraphs. Since the thinnest path problem is clearly in the NP space, this theorem establishes the NP-completeness of TP in 2-D disk hypergraphs.

\vspace{0.5em}
\begin{theorem}
MDS $\leq_P$ TP in $2$-D disk hypergrpahs.
\label{thm:NP-2ddh}
\end{theorem}
\vspace{0.5em}

To prove Theorem~\ref{thm:NP-2ddh}, consider an MDS problem in an arbitrary graph $G$. We first construct a \emph{general} directed hypergraph $H_1$ based on $G$ such that a thinnest path in $H_1$ leads to an MDS in $G$. The main challenge in the proof is to show that $H_1$ is realizable under a 2-D disk model. There are two main difficulties. First, line crossing is inevitable when we draw $H_1$ on a 2-D plane. The implementation of hyperedges that cross each other needs special care to avoid unwanted overhearing that may render the reduction invalid. Second, the geometric structure of 2-D disk hypergraphs dictates that there are at most 5 vertices (even with arbitrary ranges) that can reach a common sixth vertex but not each other. It is thus challenging to implement a vertex with up to $n$ incoming hyperedges in $H_1$ while preserving the reduction. 

Our main approach to the above difficulties is to allow \emph{directed} overhearing. Specifically, messages transmitted along one hyperedge may be heard by vertices implementing another hyperedge in $H_1$, but not vise verse. By carefully choosing the directions of the introduced overhearing, we ensure that the resulting 2-D disk hypergraph $H_2$, while having a different topological structure from $H_1$, preserves the reduction from MDS in~$G$.

Another challenge in constructing $H_2$ is to ensure the polynomial nature of the reduction. The number of additional vertices added in $H_2$ needs to be in a polynomial order with the size of $G$. This often limits using reduced transmission ranges as a way to avoid unwanted overhearing: exponentially small transmission ranges may require exponentially many vertices to connect two fixed points.     

A detailed proof is given in \ref{app:NP_2ddh}.

\vspace{\SubsectionVspace}
 \subsection{TP in $2$-D Unit Disk Hypergraphs}
We now establish the NP-completeness of TP in 2-D \emph{unit} disk hypergraphs (UDH). The proof builds upon the proof of Theorem~\ref{thm:NP-2ddh}. The only difference is that when implementing the general directed hypergraph $H_1$, we no longer have the freedom of choosing the maximum transmission range of each vertex. This presents a non-trivial challenge. As stated in Sec.~\ref{sec:npc-2ddh}, our approach to circumventing the constraints imposed by the geometrical structures of 2-D disk hypergraphs is to allow directed overhearing, which is achieved by carefully choosing different maximum transmission ranges of various vertices. To implement a 2-D UDH for the reduction, however, all vertices must have the same maximum transmission range.   

To address this issue, we introduce a special type of disk hypergraphs, called \emph{exposed disk hypergraphs}, and show that TP in $k$-D exposed disk hypergraphs can be reduced to TP in $k$-D UDH for any $k\geq 2$. We then show that the $2$-D disk hypergraph $H_2$ constructed in the proof of Theorem~\ref{thm:NP-2ddh} can be modified to an exposed hypergraph while preserving the reduction, . We thus arrive at the NP-completeness of TP in 2-D UDH based on the transitivity of polynomial time reduction.

\begin{define}
In a disk hypergraph $H=(V,E)$, let $\tau_v$ denote the closest non-neighbor \footnote{A vertex is a non-neighbor of $v$ if it is outside the maximum range $R_v$ of $v$.} of $v$. Define\footnote{The parameter $\frac{1}{2}$ can be change to an arbitrary positive value smaller than $1$.}
\begin{align*}
\epsilon_v\defeq 
\frac{1}{2}(d(v, \tau_v)-R_v),
\end{align*}
where $d(v,\tau_v)$ is the distance between $v$ and $\tau_v$ ($\epsilon_v$ is set to $1$ when $v$ does not have non-neighbors). An \emph{exposed area} $\Phi_v$ of $v$ is defined as
\begin{align*}
\Phi_v\defeq D_{v,R_v+\epsilon_v}\backslash \bigcup_{u\in V} D_{u,R_u},
\end{align*}
where $D_{v,r}$ denotes the closed ball centered at $v$ with radius $r$. A disk hypergraph is \emph{exposed} if every vertex has a non-empty exposed area.  
\end{define}
\begin{figure}[htbp]
\centering
\begin{psfrags}
\psfrag{a}[c]{$H_1$}
\psfrag{b}[c]{$H_2$}
\psfrag{c}[c]{$H_3$}
\scalefig{0.4}\epsfbox{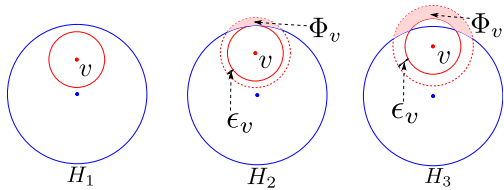}
\end{psfrags}
\caption{Exposed hypergraphs and exposed areas ($H_1$ is not exposed since $v$ has an empty exposed area; $H_2$ and $H_3$ are exposed).}
\label{fig:expose}
\end{figure}

\begin{lemma}\label{lma:expose}
TP in $k$-D exposed disk hypergraphs $\leq_P$ TP in $k$-D UDH.
\end{lemma}
\begin{proof}
The basic idea is to place super vertices at specific locations in exposed areas to force vertices on a thinnest path to use transmission ranges smaller than the maximum value. The problem is thus transformed to the case with disk hypergraphs where vertices may have different maximum transmission ranges.
A detailed proof is given in \ref{app:expose_base}.
\end{proof}

With Lemma~\ref{lma:expose} providing a bridge between disk and unit disk hypergraphs, all we need to show is that MDS can be reduced to TP in $2$-D \emph{exposed} disk hypergraphs.
\begin{lemma}
MDS $\leq_P$ TP in $2$-D exposed disk hypergrpahs.
\label{lma:NP-2dedh}
\end{lemma}
\begin{proof}
See \ref{app:expose}.
\end{proof}

Based on Lemma~\ref{lma:expose} and Lemma~\ref{lma:NP-2dedh}, we arrive at the following theorem. 
\begin{theorem}
MDS $\leq_P$ TP in $2$-D UDH.
\label{thm:NP_2dudh}
\end{theorem}

\vspace{\SubsectionVspace}
\subsection{TP in $2$-D Disk Graphs and $3$-D Unit Disk Graphs}
In this subsection, we consider the thinnest path problem in disk graphs and unit disk graphs (UDG). Recall that disk and unit disk graphs are special ring hypergraphs with $r_i=R_i$ and $r_i=R_i=1$, respectively. In other words, they can be seen as hypergraphs where each vertex has only one outgoing hyperedge directing to its prefixed neighbor set (determined by its fixed transmission power). This also shows that disk hypergraphs and disk graphs are not special cases of each other. Given the same set of vertices and their associated maximum ranges, a disk hypergraph has a different topology from a disk graph: each vertex in general has more than one outgoing hyperedges due to the freedom of using smaller transmission ranges. The same holds for UDH and UDG. As a consequence, the complexity of TP in disk and unit disk graphs cannot be inferred from Theorems \ref{thm:NP-2ddh} and \ref{thm:NP_2dudh}, and needs to be studied separately.

\begin{theorem}
MDS $\leq_P$ TP in $2$-D disk graphs.
\label{thm:NP_2ddg}
\end{theorem}
\begin{proof}
In the proof of Theorem~\ref{thm:NP-2ddh}, the vertices along the thinnest path in the constructed 2-D disk hypergraph $H_2$ all use their maximum ranges. Thus, MDS in $G$ can be reduced to TP in a disk graph constructed based on $H_2$ by including only those hyperedges associated with the maximum range of each vertex. 
\end{proof}

Next we consider TP in UDG. Unfortunately, the approach through exposed disk hypergraphs used in showing the NP-completeness of TP in UDH does not apply, since it hinges on vertices being able to use any transmission ranges smaller than a maximum value.  The difficulty, however, can be circumvented for 3-D UDG as shown in the following theorem.

\begin{theorem}
MDS in degree-$3$ graphs $\leq_P$ TP in $3$-D UDG.
\label{thm:NP_3dudg}
\end{theorem}

The proof is similar to that of Theorem~\ref{thm:NP-2ddh} with two main differences. First, line crosses are implemented by using the third dimension to ``go around'', rather than using different transmission ranges (a luxury absent in UDG) to create directed crosses. Second, reduction from MDS in graphs with a maximum degree of $3$ ensures that there are at most $4$ incoming edges to each super vertex in the reduced UDG. This makes the geometric constraint on the number (at most $11$ in a $3$-D Euclidean space) of vertices that can reach a common vertex but not each other inconsequential\footnote{We can consider a reduction from MDS in graphs with a maximum degree up to $9$ (see \ref{app:NP_3dudg}).}. A detailed proof is given in \ref{app:NP_3dudg}.

Note that using a reduction from MDS in graphs with a constant maximum degree rather than MDS in general graphs leads to a weaker statement. While MDS in both cases are NP-complete, the former is approximable with a constant ratio, and the latter a ratio of $\BigO{\log n}$. Theorems~\ref{thm:NP-2ddh}-\ref{thm:NP_2ddg} thus give a $\log n$ order lower bound on the approximation ratio of those problems whereas Theorem~\ref{thm:NP_3dudg} a constant lower bound. 

\vspace{\SectionVspace}
\section{Polynomial Complexity Problems}\label{sec:poly}
In this section, we consider the thinnest path problem in 1-D networks. We show that the problem is polynomial time by constructing an algorithm
with time complexity of $\BigO{n}$. Since the input data has size $\BigO{n}$, the proposed algorithm is order-optimal. We then consider
the $1.5$-D problem and show that the algorithm developed for 1-D networks directly applies to the $1.5$-D problem.
\vspace{\SubsectionVspace}
\subsection{1-D Networks}
Consider a network under a general propagation model with $n$ nodes located on a straight line. Each vertex $v_i$ is associated with a coordinate $x_i$ on the line (the vertex index $v_i$ and its location $x_i$ are often used interchangeably). Without loss of generality, we assume that $x_1\leq x_2\leq\ldots\leq x_n$. 

\begin{figure}[htbp]
\centering
\psfrag{s}[c]{$~~s(v_4)$}
\psfrag{t}[c]{$t(v_9)$}
\psfrag{0}[c]{$v_8$}
\psfrag{1}[c]{$v_7$}
\psfrag{2}[c]{$v_6$}
\psfrag{3}[c]{$v_5$}
\psfrag{4}[c]{$v_3~$}
\psfrag{5}[c]{$v_2$}
\psfrag{6}[c]{$v_1$}
\scalefig{0.46}\epsfbox{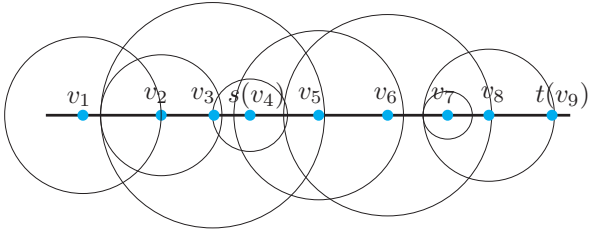}
\caption{A 1-D network (circles represent maximum ranges under a disk propagation model).}
\label{fig:1D}
\end{figure}


It is clear that every node located between the source $s$ and the destination $t$ (see Fig~\ref{fig:1D}) will hear the message
no matter which path is chosen and all nodes to the right of $t$ can be excluded from the thinnest path. Therefore, finding
the thinnest path is to minimize the number of vertices to the left of $s$ that can overhear the message. The problem is nontrivial. Due to the arbitrariness of the node locations and propagation range, a forward path (\ie every hop moves the message to the right toward $t$) from $s$ to $t$ may not exist and nodes to the
left of $s$ may need to act as relays. The question is thus how to efficiently find out whether a forward path exists and if not,
which set of nodes to the left of $s$ need to relay the message.

We propose an algorithm based on nested backward induction (NBI). For each vertex $v$,
we define its predecessor $\rho_v$ to be the rightmost vertex on the left side of $v$ that can reach $v$:
\begin{align}
\rho_v=&\arg\max_{u\in V}\{x_u: x_u<x_v, \nonumber\\
&\exists e \in E \textrm{ s.t. } s_e=u \textrm{ and } v\in T_e\}.
\end{align}
Thus, in order to reach $v$, its predecessor $\rho_v$ or a vertex to the left of $\rho_v$ has to transmit.
In other words, those vertices between $\rho_v$ and $v$ cannot directly reach $v$. Equivalently, any vertex to the right of $v$ can only hear a message from $s$ through a relay by $\rho_v$ or a vertex to the left of $\rho_v$. 

The NBI algorithm is then carried out in two steps.
In the first step, the predecessors of certain vertices are obtained one by one starting from $t$ moving toward $s$. Specifically, the predecessor of $t$, denoted by $u_1=\rho_t$, is first obtained. If $x_{u_1}\leq x_{s}$, then the first step terminates. Otherwise, the predecessor of $u_1$, denoted by $u_2=\rho_{u_1}$, is obtained and its location compared with $x_s$. The same procedure continues until the currently obtained predecessor is to the left of $s$ or is $s$ itself. The first step thus produces a sequence of vertices $u_1, u_2, \ldots, u_l$ with $u_1=\rho_t, u_2=\rho_{u_1}$, $\ldots$, $u_l=\rho_{u_{l-1}}$ and $x_{u_l}\leq x_{s}$.
Let $L_1=\{u_l,u_{l-1},\ldots,u_1,t\}$, which is a valid path from $u_l$ to $t$.
If $u_l=s$, the algorithm terminates, and the thinnest path from $s$ to $t$ is given by $L_1$.
Otherwise, we carry out Step~2 of the algorithm where we find a path from $s$ to $u_l$.
 Specifically, let $V'$ denote the set of vertices located between $u_l$ and $u_{l-1}$ including $u_l$ but not $u_{l-1}$. Let $E'$ denote the set of all hyperedges whose source and destination vertices are in $V'$. As shown in~\ref{app:F} on the correctness of the algorithm, any hyperpath $L_2$ from $s$ to $u_l$ in the sub-hypergraph $H'=(V',E')$ concatenated with $L_1$ gives a thinnest path from $s$ to $t$. Finding such an $L_2$ can be easily done by a breadth-first search (BFS) in $H'$. However, the resulting time complexity is $\BigO{n^2}$. Next we propose a special BFS procedure to reduce the time complexity to $\BigO{n}$. The trick here is to set up two pointers, $k_l$ and $k_r$, to the locations of the leftmost and the rightmost vertices in $V'$ that have been discovered. Due to the geometric structure of the 1-D network, each time we only need to search vertices to the left of $k_l$ and vertices to the right of $k_r$. The detailed algorithm is given below.

\begin{itemize}
\item[1.] Enqueue $s$, set $k_l$ and $k_r$ to the index of $s$.
\item[2.] Repeat until the queue is empty or $u_l$ is found:
\begin{itemize}
	\item Dequeue a vertex $v$ and examine it
	\item If $v=u_l$, go to step 4.
	\item Otherwise,
	\begin{itemize}
		\item While $v$ can reach $v_{k_r+1}$
		\begin{itemize}
			\item Enqueue $v_{k_r+1}$ and $k_r\leftarrow k_r+1$
			\item Set the parent of $v_{k_r+1}$ to $v$
		\end{itemize}
		\item While $v$ can reach $v_{k_l-1}$
		\begin{itemize}
			\item Enqueue $v_{k_l-1}$ and $k_l\leftarrow k_l-1$
			\item Set the parent of $v_{k_l-1}$ to $v$
		\end{itemize}	
	\end{itemize}
\end{itemize}
\item[3.] If the Queue is empty, return ``no path from $s$ to $t$''.
\item[4.] Trace back to $s$ and return $L_2$.
\end{itemize}

The following theorem establishes the correctness of the proposed NBI algorithm. Furthermore, it reveals a strong property
of the path obtained by NBI under a disk propagation model. Specifically, under a disk propagation model, we define the \emph{covered area} $A(L)$ of a hyperpath $L=\{e_1,\ldots,e_m\}$ as
\begin{equation}
  A(L)\defeq \bigcup_{i=1}^{m} D_{s_{e_i},r_{e_i}},
\label{eq:AL}
\end{equation}
where $r_{e_i}$ is the minimum transmission range that induces hyperedge $e_i$, \ie
\begin{align}
  r_{e_i}=\max_{v\in T_{e_i}}\{d(s_{e_i},v)\}.
\end{align}
Theorem~\ref{thm:NBI_property} shows that the covered area of the path obtained by NBI is a subset of the covered area of any
feasible path from $s$ to $t$.

\begin{theorem}
NBI algorithm finds the thinnest path $L^*$. Furthermore, under a disk propagation model, given any valid path $L$ from $s$ to $t$, we have $A_{L^*}\subseteq A_{L}$.
\label{thm:NBI_property}
\end{theorem}

\begin{proof}
	See \ref{app:F}.
\end{proof}

\begin{theorem}
The time complexity of the NBI algorithm is $\BigO{n}$.
\label{thm:NBI_complexity}
\end{theorem}
\begin{proof}
The $\BigO{n}$ complexity of the first step of NBI is readily seen. In the second step, the time complexity is dominated by updating the queue at each iteration. Let $k$ denote the number of iterations in step 2. Note that we only check $m_i+2$ vertices at iteration $i$, where $m_i$ is the new vertices that have been enqueued at this iteration and $\sum_{i=1}^k m_i\leq |V'|$. Also $k$ is bounded by $|V'|$. Hence the total time complexity of this step is bounded by  $\sum_{i=1}^{k} (m_i+2)\leq 3|V'|$. We thus arrive at the theorem.
\end{proof}

\vspace{\SubsectionVspace}
\subsection{$1.5$-D Networks}
We now consider the $1.5$-D problem where in-network nodes are located on a line and eavesdroppers are located in a $d$-dimensional space that contains the line network. We focus on the disk propagation model. A unit cost is incurred for each in-network node that hears the message
and a non-negative cost $c$ is incurred for each eavesdropper that hears the message.
The objective is to find a path $L^*$ from $s$ to $t$ with the minimum total cost:
\begin{align}
L^*&\defeq\arg\min_{L=\{e_1,\ldots,e_m\}}\{\sum_{v\in A(L)}c(v)\}
\end{align}
where $c(v)$ is the cost for vertex $v$, and $A(L)$ is the covered area of path $L$ as defined in
\eqref{eq:AL}.

Based on Theorem~\ref{thm:NBI_property}, it is easy to see that NBI provides the optimal solution to the 1.5-D thinnest path
problem without the knowledge of the eavesdroppers locations. More specifically,
no algorithm, even with the complete knowledge of the locations of the eavesdroppers, can obtain a thinner path than NBI
which does not require the location knowledge of the eavesdroppers.

\vspace{\SectionVspace}
\section{Approximation Algorithms}\label{sec:approx}
In this section, we introduce two approximation algorithms for the thinnest path problem and analyse their performance in different types of hypergraphs.

\vspace{\SubsectionVspace}
\subsection{Shortest Path Based Approximation Algorithm}
Given a general directed hypergraph $H$ with source vertex $s$ and destination vertex $t$, we set the weight of a hyperedge to be the number of destination vertices in this hyperedge:
\begin{align}
	w(e)\defeq |T_e| 	
\end{align} 
The shortest hyperpath algorithm from $s$ to $t$ is then obtained under this weight definition as an approximation of the thinnest path. The following theorem gives performance of this shortest path based algorithm (SPBA). 
\begin{theorem}
The SPBA algorithm provides a $\sqrt{\frac{n}{2}}$-approximation for TP in general directed hypergraphs, a $2(1+2\alpha)^d$-approximation for $d-$dimensional ring hypergraphs with $\alpha=\frac{\max_{v_i\in V} R_i }{\max\{\min_{v_i\in V} r_i,\, \min_{u,v\in V} d(u,v)\}}$. Additionally, the ratio $\sqrt{\frac{n}{2}}$ of the SPBA algorithms is asymptotically tight even in $2$-D disk hypergraphs.
\label{thm:approx}
\end{theorem}
\begin{proof}
	See \ref{app:SPBA}.
\end{proof}

\vspace{\SubsectionVspace}
\subsection{Tree Structure Based Approximation Algorithm}
Approximation occurs in two places in SPBA. First, the width of a path is approximated by the sum of the width of the hyperedges on that path. Second, the thinnest path to a vertex is assumed to go through the thinnest path to one of its incoming neighbors. The first approximation can be avoided while maintaining the polynomial nature of the approximation algorithm. In particular, we can ensure that the width of a path is correctly obtained by using the set union operation instead of summation. The assumption on the tree structure of the thinnest paths still allows us to use Dijkstra's algorithm with some modifications. Specifically, for each vertex, we need to store the current thinnest path from $s$ to this vertex rather than only the width of this path and the parent of this vertex on this path. This allows us to take the set union operation when we update the neighbors of this vertex. Given below is the performance of this tree structure based algorithm (TSBA). 
%

\begin{theorem}\label{thm:dpapprox}
The TSBA algorithm provides a $\frac{n}{2\sqrt{n-1}}$-approximation  for general directed hypergraphs, $2(1+2\alpha)^d$-approximation for $d-$dimensional ring hypergraphs with $\alpha=\frac{\max_{v_i\in V} R_i }{\max\{\min_{v_i\in V} r_i,\, \min_{u,v\in V} d(u,v)\}}$. Additionally, the ratio $\frac{n}{2\sqrt{n-1}}$ of the TSBA algorithm is tight in general directed hypergraphs and asymptotically tight in disk hypergraphs in the worst case.
\end{theorem}
\begin{proof}
See \ref{app:TSBA}.
\end{proof}

\subsection{Performance Comparison}
While the approximation ratio of TSBA is better than that of SPBA, these are worst-case performances and do not imply that TSBA outperforms SPBA in every case as shown in Fig~\ref{fig:spdp}. 
\begin{figure}[htbp]
\centering
\psfrag{v}[c]{$v$}
\psfrag{s}[c]{s}
\psfrag{t}[c]{t}
\scalefig{0.3} \epsfbox{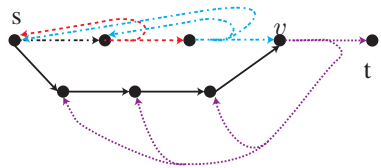}
\caption{An example where SPBA outperforms TSBA (SPBA returns the path that goes through all solid black hyperedges to $v$ and then to $t$, which is the thinnest path; TSBA returns the path that contains all colored dash hyperedges, which is not the thinnest path).}\label{fig:spdp}
\end{figure}

%
%

\begin{figure}[h!]
\centering
\psfrag{DH_1_5}[c]{}
\psfrag{SPBA}{\scalebox{.35}{SPBA}}
\psfrag{DPBA}{\scalebox{.35}{TSBA}}
\scalefig{0.4} \epsfbox{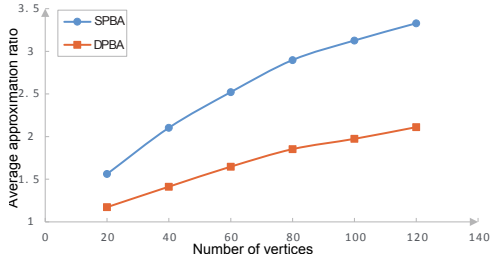}
\caption{Average performance of SPBA and TSBA (a 2-D network with $n$ vertices uniformly and randomly distributed on a $\frac{n}{\rho}\times \frac{n}{\rho}$ square with $\rho=1.5$; the maximum range of each vertex randomly chosen from interval $[R_{min}, R_{max}]$ with $R_{min}=1, R_{max}=5$; average taken over 1000 such random $2$-D disk hypergraphs).}\label{fig:simu_dh_1_5}
\end{figure}
%
%
%

Fig.~\ref{fig:simu_dh_1_5} shows the average performance of these two algorithms. We see that both algorithms have relatively small approximation ratios growing sub-linearly with the number of vertices. In general, TSBA outperforms SPBA on average, as also demonstrated in a number of other simulation results (omitted due to the space limit). 


\vspace{\SectionVspace}
\section{Conclusion}\label{sec:con}
We studied the complexity and developed optimal and approximation algorithms for the thinnest path problem for secure communications in wireless ad hoc networks. In establishing the NP-completeness of the problem, our techniques of using directed crosses and exposed disk hypergraphs may spark new tools for complexity studies in geometrical hypergraphs and graphs. The bounding techniques and the use of sphere packing results in analyzing the performance of the two approximation algorithms may also find other applications in algorithmic analysis.

\appendices

%
%
%
\vspace{\SectionVspace}
\section{Proof of Theorem~\ref{thm:NP-2ddh}}\label{app:NP_2ddh}

\subsection{Reduction from MDS to TP in A General Directed Hypergraph $H_1$}
Consider an MDS problem in a graph $G$ with $n$ vertices $v_1,\ldots,v_n$. We construct a directed hypergraph $H_1$ based on $G$ as follows. The vertex set of $H_1$ includes the $n$ vertices of $G$ augmented by a destination vertex $v_{n+1}$ and $n$ super vertices $v_1^s,\ldots,v_n^s$. A super vertex $v^s_i$ corresponds to the normal vertex $v_i$ and is a set of $n_s$ normal vertices. The hyperedges in $H_1$ are all rooted at the normal vertices $v_1,\ldots,v_n$. Specifically, rooted at $v_i$ ($1\leq i \leq n$) are  $k_i+1$ directed hyperedges, where $k_i$ is the degree of $v_i$ in $G$. Each hyperedge rooted at $v_i$ has two destinations: $v_{i+1}$ and a super vertex $v^s_j$ whose corresponding normal vertex $v_j$ dominates\footnote{A vertex in a graph is dominated by itself and any of its one-hop neighbors.} $v_i$ in the original graph $G$. Fig.~\ref{fig:RSH} is an example illustrating the construction of ${H_1}$ from $G$.

From the construction of $H_1$, we see that any path from $v_1$ to $v_{n+1}$ must traverse through all normal vertices one by one. There are multiple hyperedges leading from $v_i$ to $v_{i+1}$, each involving a super vertex that corresponds to a dominating node of $v_i$ in $G$. Thus, choosing a hyperedge going from $v_i$ to $v_{i+1}$ is equivalent to choosing a dominating node of $v_i$ in $G$. Since every path from $v_1$ to $v_{n+1}$ includes all the $n+1$ normal vertices, the thinnest path is given by the one with the minimum number of super vertices, thus leading to the MDS in $G$. At this point, the size $n_s$ of a super vertex can be any positive integer. As will become clear later, to implement $H_1$ under a 2-D disk model, additional normal vertices need to be added. As a consequence, paths from $v_1$ to $v_{n+1}$ may include different number of normal vertices. To preserve the reduction, we need to make sure that the width of a path is dominated by the number of super vertices it covers. This can be achieved by choosing an $n_s$ sufficiently large (see \ref{subsec:appA-mp}) 

\begin{figure}[htbp]
\centering
\begin{psfrags}
\psfrag{1}[c]{$v_1$}
\psfrag{2}[c]{$v_2$}
\psfrag{3}[c]{$v_3$}
\psfrag{4}[c]{$v_4$}
\psfrag{5}[c]{$v_5$}
\psfrag{t}[c]{$v_6$}
\psfrag{6}[c]{$v^s_1$}
\psfrag{7}[c]{$v^s_2$}
\psfrag{8}[c]{$v^s_3$}
\psfrag{9}[c]{$v^s_4$}
\psfrag{0}[c]{$v^s_5$}
\psfrag{a}[c]{$v_1$}
\psfrag{b}[c]{$v_2$}
\psfrag{c}[c]{$v_3$}
\psfrag{d}[c]{$v_4$}
\psfrag{e}[c]{$v_5$}
\psfrag{x}[c]{(a)}
\psfrag{y}[c]{(b)}
\scalefig{0.45}\epsfbox{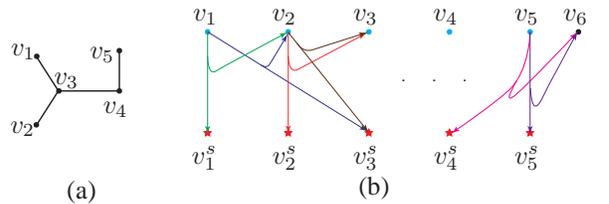}
\end{psfrags}
\caption{The construction of $H_1$ from $G$: (a) the graph $G$; (b) the hypergraph ${H_1}$ ($v_1$ is dominated by $v_1$ and $v_3$ in $G$. We thus have two hyperedges rooted at $v_1$ in $H_1$: one reaches $(v_2,v_1^s)$, the other $(v_2,v_3^s)$.).}
\label{fig:RSH}
\end{figure}
The following lemma formally establishes the correctness of the reduction.
\begin{lemma}\label{lma:gelh}
  There is a dominating set with size $k$ in $G$ if and only if there is a path from $v_1$ to $v_{n+1}$ in $H_1$ with width $kn_s+n+1$.
\end{lemma}
\begin{proof}
  First, assume that $G$ has a dominating set $S$ with size $k$. By the definition of dominating set, for each vertex $v_i$ in $G$, there is a vertex $v_j\in S$ that dominates $v_i$. From the construction of $H_1$, there exists a hyperedge $e_i$ ($i=1,\ldots,n$) in $H_1$ directing from $v_i$ to vertex $v_{i+1}$ and super vertex $v^s_{j}$ corresponding to the dominating node $v_j$ in $G$. Thus, the hyperpath $\{e_1,\ldots,e_n\}$ is a path from $v_1$ to $v_{n+1}$ with width $kn_s+n+1$. The width comes from the fact that all $n+1$ vertices in $V_{H_1}$ is on the path along with $k$ super vertices, each consisting of $n_s$ normal vertices.

 Conversely, assume a path from $v_1$ to $v_{n+1}$ in $H_1$ with width $kn_s+n+1$. Based on the construction of $H$, every path from $v_1$ to $v_{n+1}$ consists of $n$ hyperedges rooted at each of the $n$ normal vertices $v_1,\ldots,v_n$. Thus, a path with width $kn_s+n+1$ must contain $k$ super vertices. From the construction of the hyperedges, we conclude that the vertices in $G$ that correspond to those $k$ super vertices along the given path form a dominating set with size $k$.
\end{proof}

\vspace{\SubsectionVspace}
\subsection{A 2-D Grid Representation of $H_1$}

The directed hypergraph $H_1$ obtained above does not satisfy the geometric properties of 2-D disk hypergraphs. To prove Theorem~\ref{thm:NP-2ddh}, we need to modify $H_1$ to a 2-D disk hypergraph $H_2$ while preserving the reduction from MDS in $G$. Our approach is to realizing the topological structure of each hyperedge in $H_1$ by adding additional vertices with carefully chosen locations and maximum ranges to lead from the source vertex to the destination vertices of this hyperedge. The number of additional vertices, however, should be kept at a polynomial order with the problem size to ensure the polynomial nature of the reduction. This can be achieved by adding vertices on a 2-D grid with a constant grid spacing, which allows a constant maximum range thus polynomially many additional vertices. The detailed implementation of $H_1$ under a 2-D disk model is given in the next subsection. As a preparation step, we show in this subsection that the hyperedges in $H_1$ can be represented by line segments of a 2-D grid with a constant grid spacing. 

We first embed the normal vertices of $H_1$ evenly in a horizontal line in a 2-D space (see Fig.~\ref{fig:RSG_R} for an illustration). Below this line is a $2n^2\times 4n^2$ unit grid. There are $4n$ vertical lines between $v_i$ and $v_{i+1}$ ($1\leq i\leq n$) that are partitioned into three zones ($C_i^1$, $C_i^2$, $C_i^3$) of $n$, $2n$, and $n$ vertical lines, respectively. The super vertices are embedded evenly on a horizontal line below the grid. The horizontal position of super vertex $v_i^s$ is between $v_i$ and $v_{i+1}$. 

Next, we specify how a hyperedge traverses on the grid from its source vertex to its destination vertices. Recall that every hyperedge in $H_1$ directs from a normal vertex $v_i$ to a super vertex $v_j^s$ and the next normal vertex $v_{i+1}$. To preserve the reduction, we need to ensure that each hyperedge can only reach its normal vertex destination after reaching its super vertex destination. To facilitate the implementation around the super vertices (see \ref{subsubsec:super}), we designate the middle zone $C_i^2$ between $v_i$ and $v_{i+1}$ for traveling down to super vertex $v_i^s$ and then up to the corresponding normal vertex destination (see region $C_1^2$ in Fig.~\ref{fig:RSG_R}). Each hyperedge involving $v_i^s$ has two designated vertical lines in $C_i^2$ (one for going down to, the other going up from the super vertex). To connect the designated vertical lines in zone $C_i^2$ with the source vertex and then to the normal destination vertex, we designate two horizontal lines for each hyperedge. The traverse of the hyperedge completes by designating one vertical line in $C_i^1$ and one in $C_i^3$ to connect the normal vertices with the corresponding designated horizontal lines. Since there are at most $n^2$ hyperedges, the designed grid size is sufficient to ensure that each hyperedge traverses through a distinct set of line segments in the grid. . 
\begin{figure}[htbp]
\centering
\begin{psfrags}
\psfrag{1}[c]{$v_1$}
\psfrag{2}[c]{$v_2$}
\psfrag{3}[c]{$v_3$}
\psfrag{4}[c]{$v_4$}
\psfrag{5}[c]{$v_5$}
\psfrag{t}[c]{$v_6$}
\psfrag{6}[c]{$v^s_1$}
\psfrag{7}[c]{$v^s_2$}
\psfrag{8}[c]{$v^s_3$}
\psfrag{9}[c]{$v^s_4$}
\psfrag{0}[c]{$v^s_5$}
\psfrag{d}[c]{$C^2_{1}$}
\psfrag{c}[c]{$C^3_{1}$}
\psfrag{e}[c]{$C^1_{1}$}
\scalefig{0.45}\epsfbox{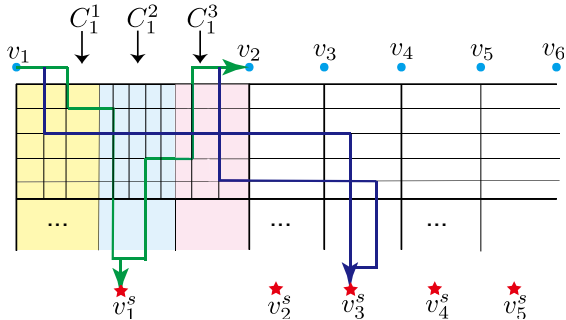}
\end{psfrags}
\caption{A 2-D grid representation of $H_1$ (the two hyperedges rooted at $v_1$ from the example given in Fig.~\ref{fig:RSH} are illustrated in green and blue, respectively).}
\label{fig:RSG_R}
\end{figure}

\vspace{\SubsectionVspace}
\subsection{Implementing $H_1$ under A 2-D Disk Model}
Based on the 2-D grid representation of $H_1$, we can construct a 2-D disk hypergraph $H_2$ that preserves the reduction. Specifically, we place a sequence of evenly spaced normal vertices with a constant maximum range along the line segments in the grid that form each hyperedge of $H_1$. The distance of two adjacent vertices is set to their maximum range. The constant maximum range can be set sufficiently small (say, $\frac{1}{5}$) to avoid overhearing across vertices on different hyperedges that may render the reduction invalid. There are two issues remain to be addressed: the implementation of crosses and that around super vertices.

\subsubsection{Implementaion of Crosses}
The line crossing in the grid representation of $H_1$ make overhearing across hyperedges inevitable. However, by exploiting the freedom of choosing the maximum range for each vertex, we can implement  \emph{directed} crosses that allow us to preserve the reduction. Specifically, when two line segments in the grid representation cross, we can choose the maximum ranges of the vertices along these two lines in such a way that messages transmitted over one line can be heard by vertices on the other but not vise verse. A specific implementation is given in Fig.~\ref{fig:cross_1}.

\begin{figure}[htbp]
\centering
\begin{psfrags}
\scalefig{0.45}\epsfbox{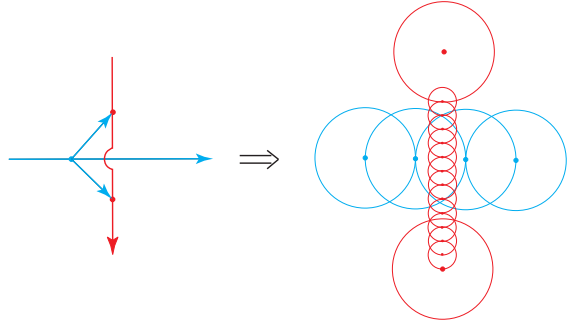}
\end{psfrags}
\caption{A disk hypergraph implementation of a directed cross where the circles represent the maximum range of vertices (messages transmitted on the blue line can be heard by nodes on the red line, but not vise verse).}
\label{fig:cross_1}
\end{figure}

Next, we show how carefully choosing the direction of each cross allows us to preserve the reduction. The cross directions are defined by assigning a level index to each line segment in the grid representation. Specifically, For a hyperedge rooted at $v_i$ in $H_1$, its line segments before and after reaching the super vertex destination have levels $i$ and $i+1$, respectively. Then, each cross has a direction pointing from the higher level segment to the lower one (\ie messages transmitted on the higher level segment can be heard by the vertices along the lower level segment but not vise verse). If the two segments have the save level, the direction of the cross can be arbitrary. To see that this directed implementation of crosses preserves the reduction, we only need to notice that any path from $v_1$ to $v_{n+1}$ still needs to go through all the $n$ normal vertices one by one and must reach one super vertex before reaching the next normal vertex. 

\subsubsection{Implementation Around Super Vertices}
\label{subsubsec:super}
Recall that a super vertex in $H_1$ is a set of $n_s$ normal vertices that have no outgoing hyperedges. It can be implemented by $n_s$ points with zero maximum range and located sufficiently close to each other (so that any path from $v_1$ to $v_{n+1}$ in $H_2$ includes either all of them or none of them).

Consider first the implementation of one incoming hyperedge to a super vertex $v_j^s$. Recall that in the 2-D grid representation of $H_1$, a hyperedge approaches and leaves $v_j^s$ through two vertical lines in zone $C_j^2$ (see Fig.~\ref{fig:RSG_R}). One implementation of this U-turn around $v_j^s$ is to add $6$ normal vertices with specific maximum ranges and locations. As shown in Fig.~\ref{fig:super_2}, these $6$ vertices include three anchor vertices $u_1^-$, $u_1^0$, and $u_1^+$ with maximum range $r$, two interface vertices $\nu_1^-$ and $\nu_1^+$ that connect with the grid, and a bridging vertex $\mu_{11}$, all with maximum range $\frac{r}{2}$. The value of $r$ and the connection with the grid will be specified later.  

\begin{figure}[htbp]
\centering
\begin{psfrags}
\psfrag{r}[c]{\small $r$}
\psfrag{a}[c]{\small $u_1^-$}
\psfrag{b}[c]{\small $u_1^0$}
\psfrag{c}[c]{\small $u_1^+$}
\psfrag{d}[c]{\small $\nu_1^-$}
\psfrag{e}[c]{\small $\nu_1^+$}
\psfrag{f}[c]{\small $\mu_{11}$}
\scalefig{0.40}\epsfbox{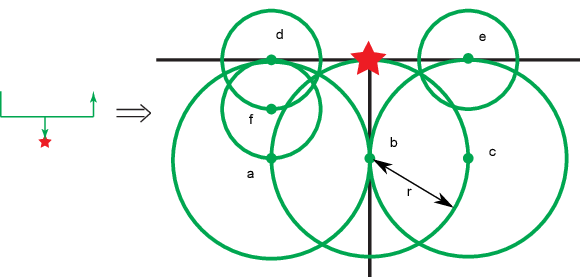}
\end{psfrags}
\caption{Implementation of one hyperedge passing through a super vertex. Starting from $\nu_1^-$, the message traverse to $\nu_1^+$ through $\mu_{11}$, $u_1^-$, $u_1^0$, $u_1^+$. The super vertex hears the message in the transmission from $u_1^0$ to $u_1^+$.}
\label{fig:super_2}
\end{figure}

The challenge remains in the implementation of up to $n$ incoming hyperedges to the same super vertex.   Note that under a 2-D disk model, one can at most have $5$ vertices (even with arbitrary ranges) that reach a common sixth vertex but not each other.
The key to circumvent this difficulty is to allow directed overhearing, similar to the idea behind the implementation of the crosses. Specifically, the reduction is preserved as long as a hyperedge rooted at $v_j$ cannot overhear a message transmitted over a hyperedge rooted at $v_i$ for any $i<j$. The detailed implementation is as follows. The fist step is to designate the vertical lines in zone $C_j^2$ to the incoming hyperedges of $v_j^s$ based on the indexes of their source vertices.   
Specifically, the incoming hyperedge with the smallest source vertex index takes the two most center lines in $C_j^2$, and so on. Consider first the implementation of the two incoming hyperedges (say, $e_1$ and $e_2$) with the smallest source vertex indexes. 
As shown in Fig.~\ref{fig:super_2_2}, we first implement $e_1$ as described above (see Fig.~\ref{fig:super_2}). The structure of the implementation of $e_2$ is similar except that the maximum range of the anchor vertices $u_2^-$, $u_2^0$, and $u_2^+$ is set to $4r$ to present unwanted overhearing. As a consequence, more bridging vertices ($\mu_{21},\mu_{22},\mu_{23}$ with maximum range $\frac{r}{2}$, $r$, and $2r$, respectively) are needed to connect the interface vertex $\nu_2^-$ to the anchor vertex $u_2^-$. Note that no vertices along $e_2$ (the centers of the blue circles in Fig.~\ref{fig:super_2_2}) are in the range of any vertices along $e_1$ (the green circles). The correct direction of overhearing is thus ensured.

The same procedure continues for any additional incoming hyperedges to $v_i^s$, in the ascending order of their source vertex indexes in $H_1$. Note that the range of the anchor vertices in the $k$th hyperedge is $4^kr$, growing exponentially with $k$. The maximum ranges (specifically, $\frac{r}{2},r,\ldots,\frac{4^k}{2}r$) of the bridging vertices $\{\mu_{ki}\}$ are chosen to preserve the polynomial nature of the reduction. In this way, the number of additional vertices for implementing the $k$th hyperedge is $2k+4$, and the total number of additional vertices around one super vertex is at most $n^2+3n$.

\begin{figure}[htbp]
\centering
\begin{psfrags}
\psfrag{a}[c]{\small $u_2^-$}
\psfrag{b}[c]{\small $u_2^0$}
\psfrag{c}[c]{\small $u_2^+$}
\psfrag{d}[c]{\small $\nu_2^-$}
\psfrag{e}[c]{\small $\mu_{21}$}
\psfrag{f}[c]{\small $\mu_{22}$}
\psfrag{g}[c]{\small $\mu_{23}$}
\psfrag{h}[c]{\small $\nu_2^+$}
\psfrag{4r}[c]{\small $4r$}
\scalefig{0.45}\epsfbox{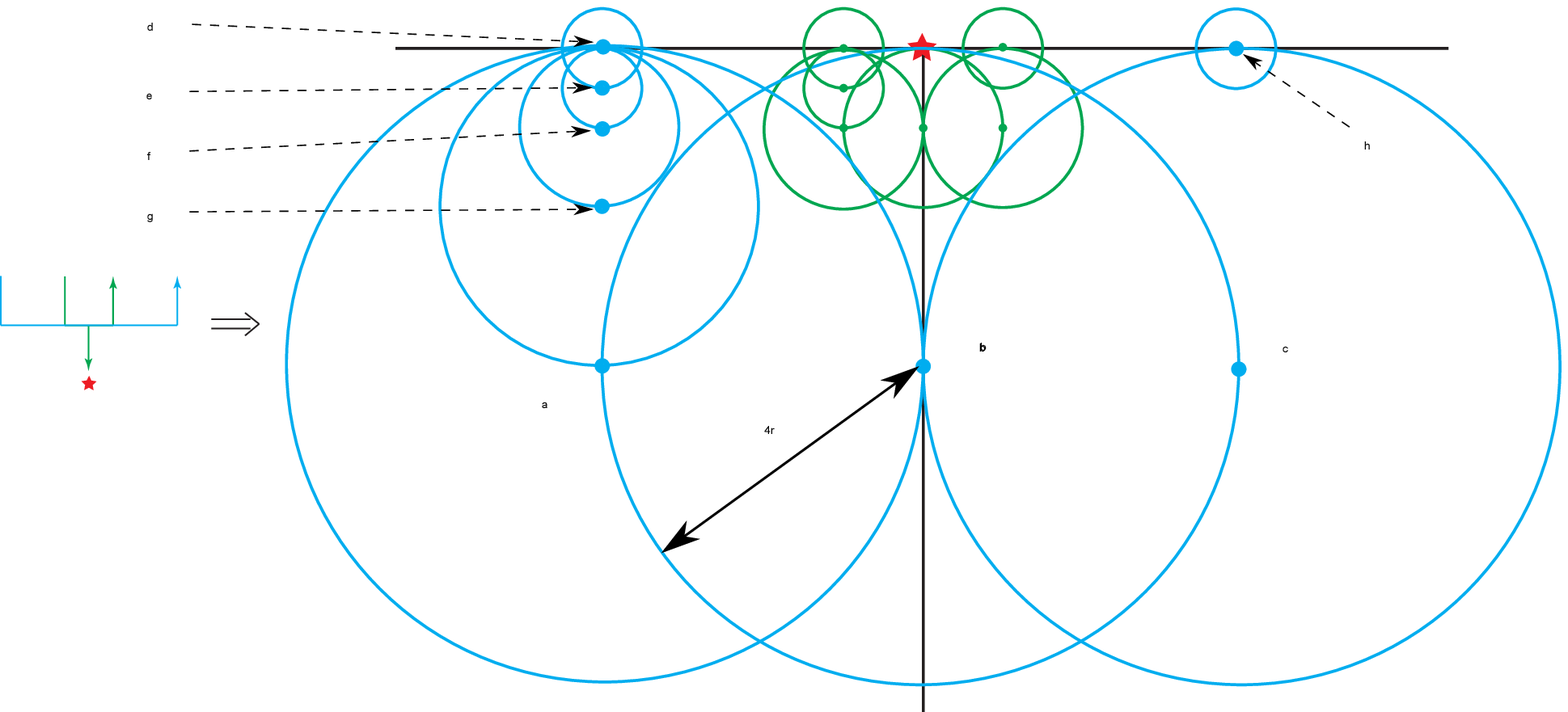}
\end{psfrags}
\caption{Implementation of the second incoming hyperedge to a super vertex.}
\label{fig:super_2_2}
\end{figure}

Next we consider the value of $r$ which should be set sufficiently small to avoid overhearing across hyperedges leading to different super vertices. Note that the width of the area covered by the additional vertices around a super vertex is $4$ times the largest maximum range of the anchor vertices. We thus set $r=4^{-n}n$, considering the distance between two adjacent super vertices being $4n$.  

The last issue is to connect the interface vertices with the grid. Each interface vertex needs to be connected with a designated vertical line in $C_j^2$. While the vertical lines in $C_j^2$ are evenly spaced, the horizontal positions of the interface vertices have an exponential structure due to the exponentially growing range of the anchor vertices. Furthermore, the vertices realizing the vertical lines in the grid have a constant range, whereas the interface vertices have an exponentially smaller range of $r=4^{-n}n$. If we connect them using a sequence of vertices with a constant range, unwanted overhearing will occur near the interface vertices. On the other hand, connecting them using vertices with range $r$ results in an exponential number of additional vertices. To preserve the correctness and the polynomial nature of the reduction, we propose the scheme detailed in Fig.~\ref{fig:gridinterfacegeo}. 
\begin{figure}[htbp]
\centering
\begin{psfrags}
\psfrag{a}[c]{\small $A$}
\psfrag{b}[c]{\small $B$}
\psfrag{c}[c]{\small $C$}
\psfrag{d}[c]{\small $D$}
\psfrag{f}[c]{\small $F$}
\psfrag{o1}[c]{\small $O_1$}
\psfrag{o2}[c]{\small $O_2$}
\psfrag{e1}[c]{\small $E_1$}
\psfrag{e2}[c]{\small $E_2$}
\psfrag{d1}{\tiny $d_1$}
\psfrag{d1'}{\tiny $d_1'$}
\psfrag{d2}{\tiny $d_2$}
\psfrag{d2'}{\tiny $d_2'$}
\psfrag{L}{\small $L$}
\psfrag{o}{\small $O$}
\psfrag{alpha}{\small $\alpha$}
\psfrag{beta}{\small $\beta$}
\psfrag{theta}{\small $\theta$}
\scalefig{0.3}\epsfbox{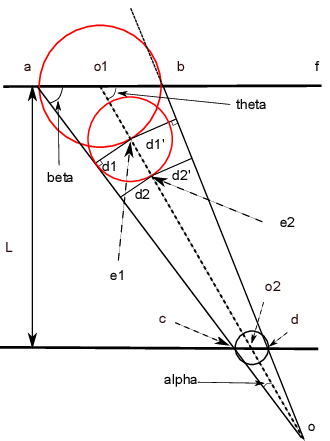}
\end{psfrags}
\caption{Consider first the downward part from the grid to a left interface vertex $\nu_i^-$. Let $O_1$ denote the location of the last vertex on the designated vertical line in the grid, and  $O_2$ the location of $\nu_i^-$. The circles centered at $O_1$ and $O_2$ represents their maximum ranges. Let $A,B$ and $C,D$ denote the intersecting points of these two circles with the horizontal lines at their centers. Let $E_1$ denote the intersection between circle $O_1$ and line $O_1O_2$. Let $d_1$ and $d_1'$ denote the distance between $E_1$ and the two lines $AC$ and $BD$, respectively. Next we draw a circle with radius $r_1=\min\{d_1,d_1'\}$ centered at $E_1$. Let $E_2$ denote the intersection between circle $E_1$ and line $O_1O_2$, and a similar circle centered at $E_2$ is drawn. This procedure is repeated to generate a sequence of circles until the last generated circle covers $O_2$. This sequence of circles $\{E_i\}$ gives the locations and the maximum ranges of the vertices connecting the grid and $\nu_i^-$. The upward part from $\nu_i^+$ to the grid is done with the same procedure except starting from $\nu_i^+$. }
\label{fig:gridinterfacegeo}
\end{figure}

Since the generated sequence of circles $\{E_i\}$ are within the boundary given by lines $AC$ and $BD$ and the boundary lines corresponding to different interface vertices do not cross (see Fig.~\ref{fig:super_3}), the above scheme does not introduce overhearing, thus preserving the reduction. The polynomial nature of the reduction can be shown  based on the following lemma.
\begin{figure}[htbp]
\centering
\begin{psfrags}
\psfrag{a}[c]{\small $C_j^2$}
\psfrag{v}[c]{\small $v_j^s$}
\scalefig{0.3}\epsfbox{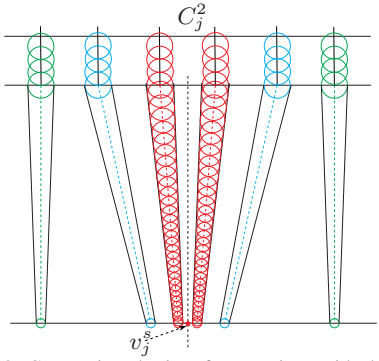}
\end{psfrags}
\caption{Connecting the interface vertices with the grid.}
\label{fig:super_3}
\end{figure}
\begin{lemma}\label{lma:super}
Consider the geometrical scheme described in Fig.~\ref{fig:gridinterfacegeo}. Assume $\angle CAB \geq\frac{\pi}{4}$. The number of circles $\{E_i\}$, denoted by $k$, satisfies $k\leq \frac{2(\log R-\log R_2)}{R_1-R_2}L+1$ when $R_1\neq R_2$, and $k\leq \frac{2L}{R_1}+1$ when $R_1=R_2$, where $R_1$, $R_2$ denote the radiuses of circles $O_1$ and $O_2$, and $L$ the distance between lines $AB$ and $CD$.
\end{lemma}

\begin{proof}
Assume first $R_1\neq R_2$. Without loss of generality, assume $R_1>R_2$. Since line $AB$ and $CD$ are parallel, the three lines $AC$, $O_1O_2$ and $BD$ intersect at one point, denoted by $O$ in Fig.~\ref{fig:gridinterfacegeo}. Let $\alpha$, $\beta$, and $\theta$ denote the angles $\angle O_1 O A$, $\angle OAB$ and $\angle OO_1B$, respectively. 

It can be shown that all the circles $\{E_i\}$ tangent with the same boundary line. Without loss of generality, assume that they tangent with line $AC$, \ie $d_i\leq d_i'$ and $r_i=d_i$. Based on simple geometry, the lengths of the line segments  of $\{OE_i\}$ forms an equal ratio sequence:
\begin{align*}
	OE_{i+1}&=OE_i-r_i=OE_i(1- \sin \alpha),
\end{align*}
with $OE_1=OO_1-R_1$. We thus have
\begin{align*}
OE_{i+1}=(OO_1-R_1)(1-\sin\alpha)^i.
\end{align*}
Based on the stopping condition of the procedure, the number $k$ of circles is given by the minimum index $i$ such that $OE_{i+1}\leq OO_2$. We thus have
\begin{align}
k&=\min \{i\in \mathbb{N}:OE_{i+1}\leq OO_2\}\nonumber\\
&=\min \{i\in \mathbb{N}:(OO_1-R)(1-\sin\alpha)^i\leq OO_2\}\nonumber\\
&=\min \{i\in \mathbb{N}:i\leq \frac{\log\frac{OO_2}{OO_1-R_1}}{\log(1-\sin\alpha)}\}\nonumber\\
&\leq \log(OO_2/OO_1)/\log(1-\sin\alpha)+1.
\label{equ:super_1}
\end{align}

Since $\Delta OO_2D$ and $\Delta OO_1B$ are similar triangles, the ratio $\frac{OO_2}{OO_1}$ equals to the ratio $\frac{R_2}{R_1}$. Also because $-\log (1-x)\geq x$ for $0\leq x\leq 1$, (\ref{equ:super_1}) can be written as
\begin{align}
k&\leq (\log R_1-\log R_2)/\sin\alpha+1.
\label{equ:super_3}
\end{align}
Because $O_1AO$ is a triangle and $\beta\geq \frac{\pi}{4}$, the value of $\sin\alpha$ can be lower bounded as the following:
\begin{align}
\sin\alpha&=\frac{R_1}{OO_1}\sin \beta\geq \frac{R_1}{OO_1}\frac{\sqrt{2}}{2}\geq \frac{R_1-R_2}{O_1O_2}\frac{\sqrt{2}}{2}.
\label{equ:super_2}
\end{align}

Furthermore, since $\theta=\alpha+\beta>\beta$, the length of $O_1O_2=\frac{L}{\sin \theta}$ has a upper bound: $O_1O_2\leq \frac{L}{\sin\beta}\leq \sqrt{2}L$. Hence (\ref{equ:super_2}) leads to:
\begin{align}
\sin\alpha&\geq (R_1-R_2)/(2L).
\label{equ:super_4}
\end{align}
Substituting (\ref{equ:super_4}) into (\ref{equ:super_3}), we have
\begin{align*}
k\leq 2L(\log R_1-\log R_2)/(R_1-R_2)+1.
\end{align*}

Consider next $R_1=R_2$. The sequence of circles $\{E_i\}$ have the same radius $R_1$. Since $O_1O_2=\frac{L}{\sin\theta}<2L$, the bound $k\leq \frac{2L}{R_1}+1$ holds.
\end{proof}

To satisfy the assumption of $\angle CAB\geq\frac{\pi}{4}$ in Lemma~\ref{lma:super}, we set the distance between the last horizontal line of the grid and the horizontal line of super vertices to $n$. This ensures the angle $\angle BAO\leq \frac{\pi}{4}$. Note that in the downward part from the grid to a left interface vertex $\nu_i^-$, $R_1$ is a constant and $R_2=\frac{4^{-n}n}{2}$. Hence the bound on $k$ given in Lemma~\ref{lma:super} can be written as:
\begin{align*}
k&\leq \frac{2(\log R_1+n\log 4-\log n)}{R_1-4^{-n}n/2}L+1\\
&\leq \frac{2(\log R_1+n\log 4)}{R_1-1/8}n+1,
\end{align*}
which is in the order of $\BigO{n^2}$. Similar argument can be made in the upward part where $R_1=4^{-n}n$ and $R_2$ is a constant. The same holds for $R_1=R_2$. Hence the total number of additional vertices to connect the grid to the interface vertices of a super vertex is in the order of $\BigO{n^3}$.

\vspace{\SubsectionVspace}
\subsection{Reduction from MDS to TP in the 2-D Disk Hypergraph $H_2$}
\label{subsec:appA-mp}
With $H_2$ constructed, we now establish the correctness of the reduction from the MDS in $G$ to the TP from $v_1$ to $v_{n+1}$ in $H_2$.
\begin{lemma}\label{lma:dh}
  Let $n_s=n_2+1$ where $n_2$ is the total number of normal vertices in $H_2$. There is a dominating set with size $k$ in $G$ if and only if there is a path from $v_1$ to $v_{n+1}$ in $H_2$ with width between $kn_s$ and $(k+1)n_s-1$. 
\end{lemma}
\begin{proof}
 The chosen value of $n_s$ ensures that the width of a path from $v_1$ to $v_{n+1}$ is dominated by the number of super vertices that it covers. The correctness of the reduction thus follows from the same arguments in the proof of Lemma~\ref{lma:gelh} based on the construction of~$H_2$.
\end{proof}
The polynomial nature of the reduction is clear from the construction of $H_2$. We thus arrive at Theorem~\ref{thm:NP-2ddh}.

\vspace{\SectionVspace}
\section{Proof of Lemma~\ref{lma:expose}}
\label{app:expose_base}
Consider a TP problem from $s$ to $t$ in a $k$-D exposed disk hypergraphs $H=(V,E)$. We construct a $k$-D UDH $H'$ as follows. First, the normal vertex set $V'$ of $H'$ is given by $V$, except that the ranges of any $v'\in V'$ equals to $\max_{v\in V} R_v$. Next, for each vertex $v'\in V'$, we place a super vertex in $\Phi_v$ (i.e., the exposed area of the corresponding vertex in $H$) that contains $|V|+1$ normal vertices located sufficiently\footnote{The $|V|+1$ normal vertices are sufficiently close such that any transmission from one of these vertices to a vertex outside this super vertex reaches all the $|V|+1$ normal vertices in this super vertex.} close to each other. The super vertices have the same range as the normal vertices in $V'$, ensuring $H'$ is a UDH. The reduction can thus be seen by noticing that while the enlarged ranges introduce additional hyperedges in $H'$, these hyperedges cannot be on a thinnest path due to the fact that they all contain at least one super vertex.

\vspace{\SectionVspace}
\section{Proof of Lemma~\ref{lma:NP-2dedh}}
\label{app:expose}
In this proof, we modify the $2$-D disk hypergraph $H_2$ in the proof of Theorem~\ref{thm:NP-2ddh} to a $2$-D exposed disk hypergraph $H_3$ while preserving the polynomial reduction. Based on the definition, a sufficient condition for a $2$-D disk hypergraph to be exposed is that none of the maximum range disks are completely inside any other. The vertices in $H_2$ for realizing the line segments of the grid satisfy this condition. We only need to modify the implementations of the crosses and around the super vertices. 

\vspace{\SubsectionVspace}
\subsection{Implementation of Crosses}
\begin{figure}[htbp1]
\centering
\begin{psfrags}
\psfrag{a}[c]{(a)}
\psfrag{b}[c]{(b)}
\psfrag{c}[c]{(c)}
\psfrag{d}[c]{(d)}
\psfrag{R}[c]{\small $R$}
\psfrag{A}[c]{\small $A$}
\psfrag{B}[c]{\small $B$}
\psfrag{C}[c]{\small $C$}
\psfrag{D}[c]{\small $D$}
\psfrag{E}[c]{\small $E$}
\psfrag{r1}{\tiny $BC=R_B=R_C=R\tan \theta \simeq 0.5543R$}
\psfrag{BD}{\tiny $BD=(1-\frac{1}{2\cos \theta})R \simeq 0.4283R$}
\psfrag{theta}{\tiny $\theta=29^o$}
\scalefig{0.4}\epsfbox{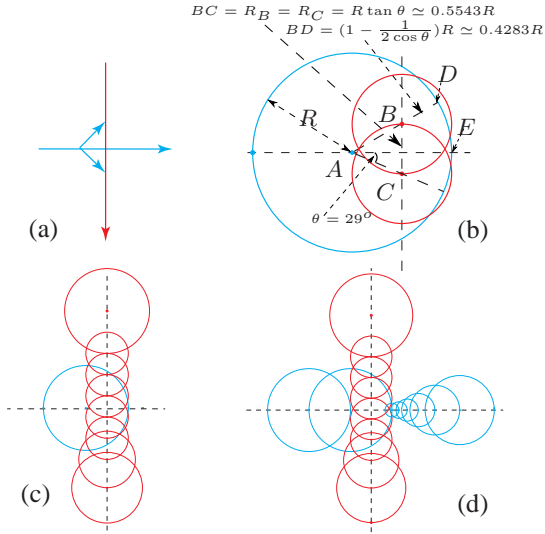}
\end{psfrags}
\caption{To implement a directed cross shown in (a), we first implement a vertex for the blue line with maximum range $R$ at location $A$ (the blue circle) shown in (b). Next we draw a perpendicular bisector between $A$ and $E$ (the right intersecting point of the circle with the line). On this vertical line, we find two points $B$ and $C$ such that $\angle BAE=\angle CAE=29^o$. At each point, we put a vertex for the red line with radius equal to the length of $BC$ (illustrated by the two red circles in (b)). Simple geometry calculation leads to $BD<BC<AB=BE$. This ensures that vertices $B$ and $C$ are exposed yet cannot overhear vertices located at $A$ and $E$. We complete the implementation by adding vertices on the vertical line $BC$ and the horizontal line $AE$ (see (c) and (d)). Note that to preserve the exposure of vertices $B$ and $C$, the maximum ranges of vertices from point $E$ to the right side need to be enlarged gradually to the constant maximum range of normal vertices on the grid (this only requires a constant number of additional vertices). } 
\label{fig:expose_cross}
\end{figure}

In the implementation of directed crosses in $H_2$ (see Fig.~\ref{fig:cross_1}), some vertices on the line with a lower level index may have an empty exposed area (see the red disks in Fig.~\ref{fig:cross_1} that are completely covered by blue ones).    To implement a direct cross in a $2$-D exposed disk hypergraph, the maximum ranges of vertices on the line with a lower level index need to be small enough to preserve the direction of the cross but also large enough to make the vertices  exposed. We propose the scheme described in Fig.~\ref{fig:expose_cross}.

\vspace{\SubsectionVspace}
\subsection{Implementation around Super Vertices }
In the previous implementation around a super vertex $v_j^s$, all the vertices are exposed except the anchor vertices $\{u_i^-,u_i^+\}$ side and the bridging vertices $\{\mu_{ik}\}$. However, we notice that these vertices would all be exposed if there were no interface vertices. Our solution is thus to move all the interface vertices  away from its original position with a constant distance and adding a constant number of vertices to connect each new interface vertex to the bridging vertex or the anchor vertex on the right side. A detailed implementation is shown in Fig.~\ref{fig:expose_super}.
\begin{figure}[htbp1]
\centering
\begin{psfrags}
\scalefig{0.4}\epsfbox{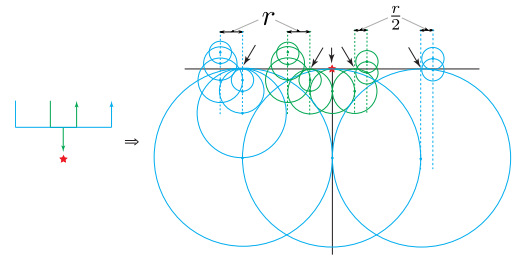}
\end{psfrags}
\caption{An interface vertex on the left side is replaced by three vertices with maximum ranges $r$, $\frac{2}{3}r$ and $\frac{r}{2}$, respectively. These three vertices are located on a vertical line to the left side of the original location of the interface vertex with a distance of $r$. An interface vertex on the right side is replaced by two vertices with maximum range $\frac{r}{2}$ located on a vertical line to the right side of the original location of the interface vertex with distance $\frac{r}{2}$. Under this implementation, the exposed areas of the anchor and bridging vertices are right above the point where they tangent with the horizontal line of the super vertices (as illustrated by the arrows). }
\label{fig:expose_super}
\end{figure}

\vspace{\SectionVspace}
\section{Proof of Theorem~\ref{thm:NP_3dudg}}\label{app:NP_3dudg}
Consider an MDS problem in a graph $G$ with a maximum degree of $3$. We first follow the first two steps in the proof of Theorem~\ref{thm:NP-2ddh} to build the grid representation of hypergraph $H_1$. Note that due to the unit range of all vertices, we set the size of the grid to a constant greater than $1$ (say, $5$) to avoid unwanted overhearing. Next, we implement this representation in a $3$-D UDG while preserving the reduction. Any line segment of hyperedges in $H_1$ is replaced by a sequence of unit disks, one tangent to another. Any cross between two line segments can be easily implemented by using the third dimension, as shown in Fig.~\ref{fig:3udg_cross}. In this implementation, there is no overhearing between vertices on these two line segments at all. Since $G$ has a maximum degree of $3$, there are at most $4$ hyperedges passing through a super vertex. It can be easily implemented without any unwanted overhearings (see Fig.~\ref{fig:3udg_super}). To prevent the super vertices from relaying messages, we place a \emph{mega} vertex besides each super vertex. This mega vertex is only within the range of this super vertex and contains more normal vertices than the total number of normal vertices in the reduced graph (including the normal vertices contained in all the super vertices but not those in other mega vertices). In this way, a path via any super vertex covers at least one mega vertex, thus cannot be the thinnest path. Fig.~\ref{fig:3udg_super} illustrates the implementation around a super vertex\footnote{We can consider reduction from MDS in graphs with a maximum degree up to $9$. In this case, there are at most $10$ incoming hyperedges. Along with the mega vertex, they can be packed around a super vertex without overhearing.}. The correctness of the reduction follows from the same arguments as in the proof of Lemma~\ref{lma:gelh}.

\begin{figure}[htbp]
\begin{psfrags}
\scalefig{0.45}\epsfbox{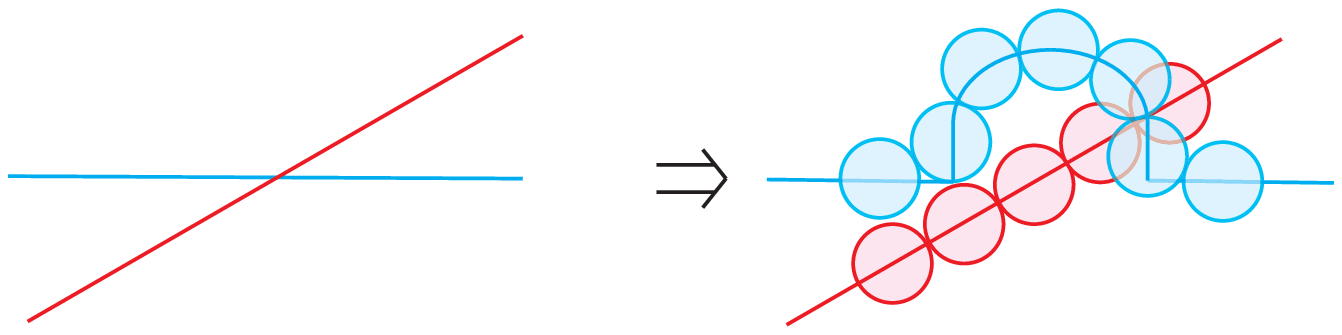}
\end{psfrags}
\caption{An implementation of a cross in $3$-D UDG.}
\label{fig:3udg_cross}
\end{figure}

\begin{figure}[htbp]
\centering
\begin{psfrags}
\psfrag{mega}[c]{\tiny Mega vertex}
\scalefig{0.4}\epsfbox{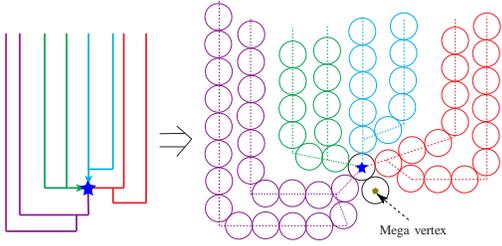}
\end{psfrags}
\caption{Implementation around the super vertices in UDG. }
\label{fig:3udg_super}
\end{figure}

\vspace{\SectionVspace}
\section{Proof of Theorem~\ref{thm:NBI_property}}\label{app:F}
We first show that as long as there exists a path from $s$ to $t$, there exists a path from $s$ to $u_{l}$ that traverses only in the sub-hypergraph $H'$. This can be shown by noticing that $u_{l}$ must hear the message from $s$ before $u_{l-1}$ and any vertex to the right of $u_{l-1}$. This is due to the monotonicity of wireless broadcast and the definition of predecessor. Consequently, there must exist a path from $s$ to $u_l$ in $H'$. 
Since $V'$ is covered by the hyperedge leading from $u_l$ to $u_{l-1}$ in $L_1$, the concatenation of $L_1$ with any path to $u_l$ in $H'$ covers the same set of vertices. Specifically, the cover of the path returned by NBI is the set of vertices located between (and including) $u_l$ and $t$. Since any path from $s$ to $t$ covers this set of vertices, the correctness of the algorithm is established. 

Next, we prove the property of $A_{L^*}$ under the disk propagation model. We first state the following lemma that follows directly from triangle inequality.
\begin{lemma}\label{lma:cir}
Let $D_1$ and $D_2$ denote two closed balls in $\mathbb{R}^d$ with radiuses $r_1$ and $r_2$, respectively. Let $a$ denote the distance between the centers of $D_1$ and $D_2$. If $0\leq a\leq |r_1-r_2|$, then $D_2\subset D_1$.
\end{lemma}

Based on Lemma~\ref{lma:cir}, for any vertex $v$ between $u_l$ and $u_{l-1}$, we have $D_{v,R_v}\subset A_{u_l,d(u_l,u_{l-1})}$. Therefore $A_{L^*}=A_{L_1}=\bigcup_{k=1}^{l} D_{u_{k},d(u_{k},u_{k-1})}$ (let $u_0=t$). Next, consider an arbitrary path $L$ from $s$ to $t$. We show that for any $u_k$ ($k=1,\ldots, l$), $D_{u_k,d(u_k,u_{k-1})}\subset A_L$. Specifically, since $u_{k-1}$ must first hear the message from $u_k$ or a vertex to the left of $u_k$, $D_{u_k,d(u_k,u_{k-1})}$ is a subset of the covered area of this hop in $L$ based on Lemma~\ref{lma:cir}. This completes the proof.

\vspace{\SectionVspace}
\section{Proof of Theorem~\ref{thm:approx}}\label{app:SPBA}
\subsection{For General Directed Hypergraphs}

Let $L_1$ denote the path from $s$ to $t$ provided by SPBA and $L_{opt}=\{e_1,e_2,\ldots,e_k\}$ the thinnest path. If multiple thinnest path exists, let $L_{opt}$ be the one with the minimum number of hyperedges. Let $\mathcal{L}(L)$ denote the length (\ie the sum of hyperedge weights) of $L$.

Since each vertex covered in $L_1$ (except the source $s$) contributes to the weight of at least one hyperedge in $L_1$, the width $W(L_1)$ is no larger than the length of this path plus one. Also because $L_1$ is the shortest path, its length is no larger than the length of $L_{opt}$. We thus have
\begin{align}
W(L_1)\leq \mathcal{L}(L_1)+1\leq \mathcal{L}(L_{opt})+1.
\label{equ:approx_spba_basic}
\end{align}
We then obtain the approximation ratio by deriving an upper bound of $\mathcal{L}(L_{opt})$ as a function of $W(L_{opt})$.

Note that the destination set $T_e$ of hyperedge $e_i$ on $L_{opt}$ cannot contain $k-i+1$ vertices: its own source vertex $s_{e_i}$ and vertices in $\{s_{e_{i+2}},s_{e_{i+3}},\ldots,s_{e_k},t\}$. The later holds because otherwise $L_{opt}$ is not the the thinnest path with minimum number of hyperedges. We thus have
\begin{align}
\mathcal{L}(L_{opt})&\leq \sum_{i=1}^{k} (W(L_{opt})-(k-i+1))\nonumber\\
&= kW(L_{opt})-k(k+1)/2\nonumber\\
&\leq W(L_{opt})(W(L_{opt})-1)/2,
\label{equ:approx_spba_1}
\end{align}
where (\ref{equ:approx_spba_1}) comes from $k\leq W(L_{opt})-1$. Substituting (\ref{equ:approx_spba_1}) into (\ref{equ:approx_spba_basic}), we have
\begin{align}
W(L_1)&\leq W(L_{opt})(W(L_{opt})-1)/2+1\nonumber\\
&\leq W^2(L_{opt})/2,
\label{equ:approx_spba_2}
\end{align}
where (\ref{equ:approx_spba_2}) holds since $W(L_{opt})\geq 2$.

Based on (\ref{equ:approx_spba_2}), if $W(L_{opt})\leq \sqrt{2n}$, then $W(L_1)\leq \frac{1}{2}W^2(L_{opt})\leq \sqrt{\frac{n}{2}}W(L_{opt})$. Otherwise, we have $W(L_1)\leq n\leq \sqrt{\frac{n}{2}}W(L_{opt})$. In summary, SPBA provides a $\sqrt{\frac{n}{2}}$ approximation.

\vspace{\SubsectionVspace}
\subsection{For Ring Hypergraphs}
Since a ring hypergraph is a special directed hypergraph, all the analysis in the previous subsection applies. Specifically, inequality (\ref{equ:approx_spba_basic}) holds. The problem then remains on obtaining a tighter upper bound of $\mathcal{L}(L_{opt})$ based on the geometrical properties of ring hypergraphs.

First, note that the length of a hyperpath $L$ equals to the sum of the number of times each vertex is reached. Let $E_v$ denote the set of hyperedges on $L_{opt}$ that include $v$ in their destination sets, \ie
\begin{align*}
  E_v&\defeq\{e\in L_{opt}:v\in T_e\}.
\end{align*}
Now we construct a subset $E'_v$ of $E_v$ by iteratively removing one from any pair of hyperedges whose positions in $L_{opt}$ are adjacent until no such pair exists. Because at most half of the hyperedges are removed from $E_v$, the size of $E'_v$ is at least half of the size of $E_v$, in another word $|E_v|\leq 2|E'_v|$.

Let $R_{max}$ and $R_{min}$ denote the largest maximum range and the smallest minimum range among all vertices in the given ring hypergraph $H_r$, respectively. Let $R'_{min}$ be the larger one between $R_{min}$ and the smallest distance between any two vertices in $H_r$. Based on the construction of $E'_v$, the set of source vertices of hyperedges in $E'_v$ satisfies two properties. First, based on the definition of ring hypergraphs, the distance between any source vertex in the set and $v$ is no larger than the maximum range of this vertex and hence no larger than $R_{max}$. Second, the distances between any two source vertices in the set are larger than $R_{min}$ and hence $R'_{min}$. Otherwise the two hyperedges rooted at these two vertices can reach the source vertex of each other and hence they are adjacent in $L_{opt}$ (recall that $e_i\in L_{opt}$ cannot reach any vertex in $\{s_{e_{i+2}},s_{e_{i+3}},\ldots,s_{e_k}\}$).

Given these two properties, the size of $E'_v$ thus is upper bounded by the maximum number of points in the Euclidean space that are at most $R_{max}$ away from $v$ and at least $R'_{min}$ apart from each other. This is equivalent to a sphere packing problem of arranging the maximum number of small spheres with radius ${R'_{min}}/{2}$ inside a large sphere with radius $R_{max}+{R'_{min}}/{2}$. An upper bound of this packing problem is the ratio between the volumes of the large and small spheres. We thus have:
\begin{align*}
|E'_v|\leq \frac{(R_{max}+R'_{min}/2)^d}{(R'_{min}/2)^d}=(1+2\alpha)^d,
\end{align*}
where $\alpha={R_{max}}/{R'_{min}}$. Recall that $|E_v|\leq 2|E'_v|$. Note that the destination $t$ can only be reached by the last hyperedge $e_k$ and hence $|E_t|=|\{e_k\}|=1$. We thus have
\begin{align}
  \mathcal{L}(L_{opt})&=\sum_{v\in \widehat{L}_{opt}\backslash \{t\}} |E_v|+|E_t|\\
  &\leq 2(1+2\alpha)^d (W(L_{opt})-1)+1\nonumber\\
  &\leq 2(1+2\alpha)^d W(L_{opt})-1.
  \label{equ:approx_spba_ring_2}
\end{align}
Substituting (\ref{equ:approx_spba_ring_2}) into (\ref{equ:approx_spba_basic}). we have
\begin{align}
  W(L_1)&\leq \mathcal{L}(L_{opt})+1\leq 2(1+2\alpha)^d W(L_{opt}),\label{equ:sp_ring_ratio}
\end{align}
\ie SPBA provides a $2(1+2\alpha)^d$-approximation for TP in $d$-D ring hypergrpahs.

\vspace{\SubsectionVspace}
\subsection{Asymptotic tightness}
We now prove that $\sqrt{\frac{n}{2}}$-ratio is asymptotically tight even for $2$-D disk hypergrpahs. The proof has two steps. First, we construct a directed hypergraph $H$ for which the worst case ratio is asymptotically reached. Next, we show a 2-D disk implementation of $H$.

Consider the the following hypergraph $H$ illustrated in Fig.~\ref{fig:example_spba} with $k$ red vertices $v_1,\ldots,v_{k}$ and $k'$ blue vertices $u_1,\ldots,u_{k'}$ along with the source $s$ and the destination $t$. Each red vertex $v_i$ has one outgoing hyperedge $e$ with $T_{e}=\{v_1,\ldots,v_{i-1},v_{i+1}\}$ (let $v_{k+1}$ denote $t$). Each blue $u_i$ has one outgoing hyperedge~$e$ with $T_{e}=\{u_{i+1}\}$ (let $u_{k'+1}$ denote $t$). Finally, we add two hyperedges that connect source $s$ to $v_1$ and $u_1$ respectively. 

\begin{figure}[htbp]
\centering
\psfrag{s}[c]{$s$}
\psfrag{t}[c]{$t$}
\psfrag{v1}[c]{$v_1$}
\psfrag{v2}[c]{$v_2$}
\psfrag{v3}[c]{$v_3$}
\psfrag{vk}[c]{$v_k$}
\psfrag{u1}[c]{$u_1$}
\psfrag{u2}[c]{$u_2$}
\psfrag{u3}[c]{$u_3$}
\psfrag{uk}[c]{$u_{k'}$}
\scalefig{0.3} \epsfbox{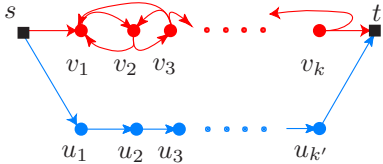}
\caption{A worst case scenario for SPBA.}\label{fig:example_spba}
\end{figure}

Let $k'=k(k+1)/2+1$. Since the shortest path traverses  through the blue hyperedges while the thinnest path through the red ones, the approximation ratio is given by:
\begin{align}
\gamma(k)={(k^2+k+2)}/{(2k+4)}.
\label{eq:spba}
\end{align}
Note that the total number of vertices is
\begin{align*}
n=k+k'+2=k+k(k+1)/2+1.
\end{align*}
When $n$ is large, $k\sim\sqrt{2n}$ and $\alpha(k)\sim\frac{k}{2}\sim\sqrt{\frac{n}{2}}$.

Next, we implement the above hypergraph under a 2-D disk model as illustrated in Fig.~\ref{fig:spdhexample}. The red vertices are located on a straight line with $R_{v_1}=R_{v_2}=1,\,R_{v_i}=2^{i-2}$ for $i>2$. The source vertex $s$ is located on the line to the left of $v_1$, and both its maximum range and its distance to $v_1$ equals to $R_{v_k}$.  The terminal vertex $t$ has a maximum range of $0$ and is located to the right of $v_k$ with a distance of $R_{v_k}$. The maximum range of a blue vertex $u_i$ is $R_{v_k}-i \epsilon$ where $\epsilon$ is a small positive value to prevent $u_{i-1}$ from  overhearing messages transmitted by $u_i$. And the blue vertices are located on a route from $s$ to $t$ that contains two vertical line segments of length $(1+l)R_{v_k}$ and a horizontal one of length $3R_{v_k}$, as demonstrated by the blue dash lines in Fig.~\ref{fig:spdhexample}. The positive parameter $l$ is used to prevent a blue vertex from  overhearing the last red vertex $v_k$. In the asymptotic regime with large $k$, $l$ can be set sufficiently large so that the $k'$ blue vertices can be implemented along the depicted route from $s$ to $t$. 

\begin{figure}[htbp]
\centering
\psfrag{s}[c]{$s$}
\psfrag{t}[c]{$t$}
\psfrag{R1}[c]{$(1+l)R_{v_k}$}
\psfrag{R2}[c]{$3R_{v_k}$}
\scalefig{0.3} \epsfbox{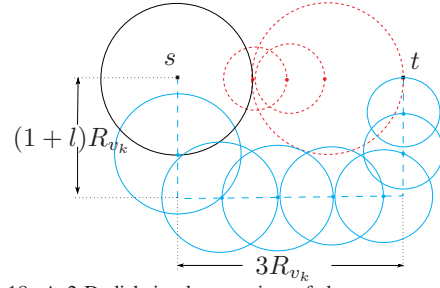}
\caption{A 2-D disk implementation of the worst case scenario for SPBA.}\label{fig:spdhexample}
\end{figure}

\vspace{\SectionVspace}
\section{Proof of Theorem~\ref{thm:dpapprox}}\label{app:TSBA}
Let $L_2$ denote the path in hypergraph $H$ from $s$ to $t$ given by the TSBA algortihm and $L_{opt}$ the thinnest path. Let $L_1(v)$ and $L_2(v)$ denote the paths from $s$ to a vertex $v$ given by SPBA and TSBA, respectively. The following lemma establishes a property of $L_2(v)$.
\begin{lemma}
\label{lma:dp_property}
For any hyperedge $e$ in $H$, we have, $\forall v\in T_e$, 
\begin{align*}
W(L_2(v))\leq |\widehat{L}_2(s_e)\cup T_e|.
\end{align*}
\end{lemma}
\begin{proof}
Lemma~\ref{lma:dp_property} follows directly from the tree structure of TSBA. 
\end{proof}

\subsection{For General Directed Hypergraphs}

Let $L_{opt}=\{e_1,\ldots,e_k\}$ denote the thinnest path.
For the ease of presentation, let the sequence of source vertices $s_{e_1},\ldots,s_{e_k}$ and the final destination $t$ be denoted as $v_1,\ldots,v_{k+1}$. 
Let $U=\{v_i\}_{i=1}^{k+1}$.
Based on Lemma~\ref{lma:dp_property}, we have, for all $i=1,\ldots,k$,
\begin{align}
&W(L_2(v_{i+1}))\leq |\widehat{L}_2(v_i)\cup T_{e_i}|\nonumber\\
&\leq W(L_2(v_i))+|T_{e_i}\backslash\{v_1,\ldots,v_{i+1}\}|+1\nonumber\\
&= W(L_2(v_i))+|T_{e_i}\backslash U|+1,
\label{equ:approx_TSBA_1}
\end{align}
where (\ref{equ:approx_TSBA_1}) holds since $T_{e_i}$ does not contain vertices in $\{s_{e_{i+2}},\ldots,s_{e_{k}},t\}$. Summing (\ref{equ:approx_TSBA_1}) over $i$, and noticing $W(L_2(v_1))=1$ and $L_2(v_{k+1})=L_2$, we have:
\begin{align*}
W(L_2)&\leq k+1+\sum_{i=1}^k |T_{e_i}\backslash U|.
\end{align*}
Next, since
\begin{align*}
W(L_{opt})&=k+1+|\cup_{i=1}^k (T_{e_i}\backslash U)|,
\end{align*}
we can upper bound $|T_{e_i}\backslash U|$ by $W(L_{opt})-k-1$ for any $i=1,\ldots,k$. Thus
\begin{align*}
W(L_2)\leq k+1+k(W(L_{opt})-k-1).
\end{align*}
The right side of this inequality is a quadratic function of $k$ with the maximum at $k=W^2(L_{opt})/2$. We thus have
\begin{align*}
W(L_2)&\leq 1+W^2(L_{opt})/4.
\end{align*}

If $W(L_{opt})\leq 2\sqrt{n-1}$, the approximation ratio $\gamma$ is given by
\begin{align}
\gamma&= W(L_{opt})/{4}+{1}/{W(L_{opt})}\leq {n}/{2\sqrt{n-1}}.\label{equ:dp_ratio2_part1}
\end{align}
The inequality holds because the function $\frac{x}{4}+\frac{1}{x}$ is an increasing function for $x\geq 2$.

If $W(L_{opt})>2\sqrt{n-1}$, we have
\begin{align}
\gamma\leq {n}/{W(L_{opt})}\leq {n}/{(2\sqrt{n-1})}.\label{equ:dp_ratio2_part2}
\end{align}

This completes the proof for case of general directed hypergraphs.

\subsection{For Ring Hypergraphs}

Let $L_1(t)=\{e_1,e_2,\ldots,e_k\}$ be the shortest path from $s$ to $t$. Let $v_i$ denote the source vertex of $e_i$ and $v_{k+1}=t$. We prove, through induction, the following inequality for all $i=1,\ldots,k+1$:
\begin{align}
W(L_2(v_i))\leq \mathcal{L}(L_1(v_i))+1.
\label{equ:lma_ratio2}
\end{align}

When $i=1$, (\ref{equ:lma_ratio2}) holds since
\begin{align*}
W(L_2(v_1))=1,\mathcal{L}(L_1(v_1))=0.
\end{align*}

Now assume that (\ref{equ:lma_ratio2}) holds for $i-1$, \ie $W(L_2(v_{i-1}))\leq \mathcal{L}(L_1(v_{i-1}))+1$. Based on Lemma~\ref{lma:dp_property} and this induction assumption, we have:
\begin{align*}
W(L_2(v_i))&\leq |\widehat{L}_2(v_{i-1})\cup T_{e_{i-1}}|\\
&\leq W(L_2(v_{i-1}))+|T_{e_{i-1}}| \\
&\leq  \mathcal{L}(L_1(v_{i-1}))+1+|T_{e_{i-1}}|\\
&=\mathcal{L}(L_1(v_i))+1.
\end{align*}
This completes the induction. Considering $v_{k+1}=t$, we have
\begin{align}
W(L_2)\leq \mathcal{L}(L_1(t))+1.
\label{equ:approx_tsba_ring}
\end{align}

From (\ref{equ:approx_tsba_ring}) and (\ref{equ:approx_spba_basic},\ref{equ:sp_ring_ratio}), we have $W(L_2(t))\leq 2(1+2\alpha)^d W(L_{opt})$, \ie TSBA provides a $2(1+2\alpha)^2$ approximation for ring hypergrpahs.

\vspace{\SubsectionVspace}
\subsection{Asymptotic tightness}
We first construct a directed hypergraph $H$ as illustrated in Fig.~\ref{fig:dpexample1}. The vertex set of $H$ consists of two types of vertices: $k+1$ normal vertices $v_0,\ldots,v_{k-1}$ and $t$, and $k$ super vertices $u_1,\ldots,u_{k-1},u$, each containing $k-1$ normal vertices. Rooted at each normal vertex $v_i$ are two hyperedges $e_{i+1}$ and $e_{i+1}'$. Hyperedge $e_i$ has destination vertices $T_{e_i}=\{v_i,u_i\}$ and hyperedge $e_i'$ has destination vertices $T_{e_i'}=\{v_i,u\}$.

It is easy to see that the thinnest path from $v_0$ to $t$ is $L_{opt}=\{e'_1,e'_2,\ldots,e'_k,e'_{k}\}$ with width $2k$. However, TSBA returns the path $L^d=\{e_1,e_2,\ldots,e_k\}$ with width $k^2+1$ in the worst case\footnote{Note that $v_{i}$ can update $v_{i+1}$ through both $e_{i+1}$ and $e_{i+1}'$ with the same width. Since the order of hyperedges used in the update is arbitrary, $e_{i+1}$ could be used to update $v_{i+1}$ for all $1\leq i\leq k-1$ in the worst case.}. The approximation ratio is
\begin{align*}
\gamma&=\frac{W(L)}{W(L_{opt})}=\frac{k^2+1}{2k}=\frac{n}{2\sqrt{n-1}}.
\end{align*}

\begin{figure}[htbp]
\centering
\psfrag{1}[c]{$v_1$}
\psfrag{2}[c]{$v_2$}
\psfrag{s}[c]{$v_0$}
\psfrag{t}[c]{$t$}
\psfrag{l}[c]{$v_{k-2}$}
\psfrag{k}[c]{$v_{k-1}$}
\psfrag{u1}[c]{$u_1$}
\psfrag{u2}[c]{$u_2$}
\psfrag{uk}[c]{$u_{k-1}$}
\psfrag{m}[c]{$u$}
\psfrag{e1}[c]{$e_1$}
\psfrag{e2}[c]{$e_2$}
\psfrag{ek}[c]{$e_{k-1}$}
\psfrag{e21}[c]{$e'_1$}
\psfrag{e22}[c]{$e'_2$}
\psfrag{e2k}[c]{$e'_{k-1}$}
\psfrag{ek+1}[c]{$e'_{k}$}
\scalefig{0.45} \epsfbox{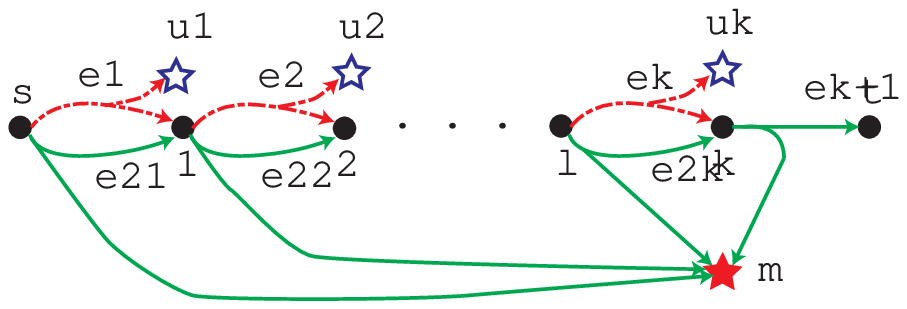}
\caption{A worst case scenario for TSBA.}\label{fig:dpexample1}
\end{figure}

Given the similarity between $H$ and the hypergraph $H_1$ constructed in the proof of Theorem~\ref{thm:NP-2ddh}, we can follow the same approach given in \ref{app:NP_2ddh} to implement $H$ under a 2-D disk model. However, this implementation requires adding additional vertices (referred to as auxiliary vertices) that may render our previous approximation analysis invalid. To maintain the ratio, each original vertex (including the vertices in a super vertex) in $H$ is replaced with $c$ vertices (clustered together) in its 2-D disk implementation, where $c$ is the number of auxiliary vertices introduced by the implementation. In this case, TSBA returns a path that covers $\{u_1,\ldots,u_{k-1},v_0,\ldots,v_k\}$ along with a set of auxiliary vertices. The thinnest path covers $\{u,v_0,\ldots,v_k\}$ and another set of auxiliary vertices. The approximation ratio in this $2$-D disk hypergraph is given by
\begin{align*}
\gamma&=\frac{(k^2+1)c+c'}{2kc+c''}=\frac{k^2+1+\frac{c'}{c}}{2k+\frac{c''}{c}}
\end{align*}
where $c'$ and $c''$ denote the number of auxiliary  vertices covered by the path returned by TSBA and the thinnest path. Since $\frac{c'}{c}\leq 1$ and $\frac{c''}{c}\leq 1$, when $n$ is large, we have $\gamma\sim \frac{n}{2\sqrt{n-1}}$, \ie the approximation ratio $\frac{n}{2\sqrt{n-1}}$ is asymptotically tight in $2$-D disk hypergraphs.


\bibliographystyle{ieeetr}

{\footnotesize
\begin{thebibliography}{10}

\bibitem{berge1976graphs}
C.~Berge, \emph{Graphs and hypergraphs}.\hskip 1em plus 0.5em minus 0.4em\relax
  North-Holland Pub. Co., 1976.

\bibitem{Ramanathan&Etal:11INFORCOM}
R.~Ramanathan, A.~Bar-Noy, P.~Basu, M.~Johnson, W.~Ren, A.~Swami,
Q.~Zhao, ``Beyond graphs: Capturing groups in networks,''in Computer Communications Workshops (INFOCOM WKSHPS), 2011 IEEE Conference on, pp. 870-875.



\bibitem{Gallo&Etal:93DAM}
G.~Gallo, G.~Longo, S.~Nguyen, and S.~Pallottino, ``Directed hypergraphs and
  applications,'' \emph{Discrete Applied Mathematics}, vol.~42, no. 2-3, pp. 177--201,
  April 1993.

\bibitem{Anjum&Mouchtaris:book}
F. Anjum and P. Mouchtaris, \emph{Security for Wireless Ad Hoc Networks,} Wiley, 2007.


\bibitem{Hu04}
Y.~Hu, A.~Perrig, ``A Survey of Secure Wireless Ad Hoc Routing,'' \emph{Ieee Securirty and Privacy Magazing}, 2(3):28-39, June 2004.


\bibitem{Chechik&etal2012}
S.Chechik, M. Johnson, M. Parter, D. Peleg, ``Secluded Connectivity Problems'', available online at http://arxiv.org/pdf/1212.6176v1.pdf

\bibitem{ausiello1992optimal}
G.~Ausiello, G.~Italiano, and U.~Nanni, ``Optimal traversal of directed
  hypergraphs,'' \emph{International Computer Science Institute}, 1992.



\bibitem{Knuth:77IPL}
D.~E. Knuth, ``A generalization of dijkstra's algorithm,'' \emph{Information
  Processing Letters}, vol.~6, no.~1, pp. 177--201, February 1977.
  


\bibitem{ausiello1990dynamic}
G.~Ausiello, U. Nanni, and G.F. Italiano, ``Dynamic maintenance of directed
  hypergraphs,'' \emph{Theoretical Computer Science}, vol.~72, no. 2-3, pp.
  97--117, 1990.



\bibitem{Gao&etal:12WiOpt}
J.~Gao, Q.~Zhao, W.~Ren, A.~Swami, R.~Ramanathan, A.~Bar-Noy ``Dynamic Shortest Path Algorithms for Hypergraphs,'' in
\emph{Proc. of the 10th International Symposium on Modeling and Optimization in Mobile, Ad Hoc, and Wireless Networks (WiOpt)}, May, 2012.

\bibitem{Hu1961}
T.C.Hu, ``The Maximum Capacity Route Problem,'' in \emph{Operations Research}, vol. 9, No. 6, pp. 898--900, 1961

\bibitem{Punnen91}
A.P.Punnen, ``A linear time algorithm for the maximum capacity path problem,'' in \emph{European Journal of Operational Research}, volume 53, issue 3, 15 August 1991, pages 402–404



\bibitem{Garey&Johnson1979}
M.R.~Garey, D.S.~Johnson \emph{Computers and Intractability: A guide to the theory of NP-completeness,} Freeman, San Francisco, 1978.




\end{thebibliography}
}

\end{document}